\documentclass[mnsc]{informs3aa} 

\usepackage[final]{showkeys}
\usepackage{subfigure}

\OneAndAHalfSpacedXI 

\usepackage{endnotes}
\let\footnote=\endnote

\usepackage{longtable}
\usepackage{tabularx}
\usepackage{booktabs}
\usepackage{etex}
\usepackage{multirow}
\usepackage{array}
\usepackage{bbm}

\usepackage[title]{appendix}

\newtheorem{fact}{Fact}

\newcommand{\ud}{\mathrm{d}}

\newcommand{\regret}{\mathrm{Regret}}

\renewcommand{\hat}{\widehat}
\renewcommand{\tilde}{\widetilde}

\usepackage{amsbsy, amsfonts, amsgen, amsmath, amsopn, amssymb, amstext,
amsxtra, bezier, color, enumerate, graphicx, latexsym, verbatim,
pictexwd, supertabular, url, dsfont,leftidx, mathrsfs, appendix, setspace}
\usepackage{algorithm}
\usepackage{algorithmicx}
\usepackage{algpseudocode}
\makeatletter
\def\BState{\State\hskip-\ALG@thistlm}
\makeatother

\usepackage{natbib}
 \bibpunct[, ]{(}{)}{,}{a}{}{,}%
 \def\bibfont{\small}%
\usepackage[colorlinks=true,bookmarks=false,urlcolor=blue, citecolor=blue,linkcolor=blue,bookmarksopen=false,draft=false]{hyperref}

\usepackage{xcolor}
\usepackage{epstopdf}

\renewcommand{\hat}{\widehat}
\renewcommand{\tilde}{\widetilde}
\renewcommand{\bar}{\overline}

\definecolor{DSgray}{cmyk}{0,1,0,0}

\TheoremsNumberedThrough     
\ECRepeatTheorems

\EquationsNumberedThrough    

\MANUSCRIPTNO{} 

\begin{document}



\RUNTITLE{Privacy-Preserving Dynamic Personalized Pricing}

\TITLE{Privacy-Preserving Dynamic Personalized Pricing with Demand Learning}

\ARTICLEAUTHORS{%
\AUTHOR{Xi Chen\thanks{Author names listed in alphabetical order.}}
\AFF{Leonard N.~Stern School of Business, New York University, \EMAIL{xc13@stern.nyu.edu}}
\AUTHOR{David Simchi-Levi}
\AFF{Institute for Data, Systems, and Society, Department of Civil and Environmental Engineering and Operations Research Center, Massachusetts Institute of Technology, \EMAIL{dslevi@mit.edu}}
\AUTHOR{Yining Wang}
\AFF{Warrington College of Business, University of Florida, \EMAIL{yining.wang@warrington.ufl.edu}}
} 

\ABSTRACT{The prevalence of e-commerce has made customers' detailed personal information readily accessible to retailers, and this information has been widely used in pricing decisions. When using personalized information, the question of how to protect the privacy of such information becomes a critical issue in practice. In this paper, we consider a dynamic pricing problem over $T$ time periods with an \emph{unknown} demand function of posted price and personalized information. At each time $t$, the retailer observes an arriving customer's personal information and offers a price. The customer then makes the purchase decision, which will be utilized by the retailer to learn the underlying demand function. There is potentially a serious privacy concern during this process: a third-party agent might infer the personalized information and purchase decisions from price changes in the pricing system. Using the fundamental framework of differential privacy from computer science, we develop a privacy-preserving dynamic pricing policy,  which tries to maximize the retailer revenue while avoiding information leakage of individual customer's information and purchasing decisions. To this end, we first introduce a notion of \emph{anticipating} $(\varepsilon, \delta)$-differential privacy that is tailored to  the dynamic pricing problem. Our policy achieves both the privacy guarantee and the performance guarantee in terms of regret. Roughly speaking, for $d$-dimensional personalized information, our algorithm achieves the expected regret at the order of $\tilde{O}(\varepsilon^{-1} \sqrt{d^3 T})$, when the customers' information is adversarially chosen. For stochastic personalized information, the regret bound can be further improved to $\tilde{O}(\sqrt{d^2T} + \varepsilon^{-2} d^2)$.
}



\date{}

\maketitle

\parskip=6pt

\section{Introduction}

The increasing prominence of e-commerce has given retailers an unprecedented power to understand customers as individuals and to tailor their services accordingly. For example, personal information is known to be used in pricing on travel websites \citep{Hannak14Measuring} and Amazon \citep{Chen2016empirical}; \cite{Linden2003} illustrates how personal information is used in Amazon recommender systems to achieve a dramatic increase in click-through and conversion rates.  Although personalized pricing may involve complicated legal issues in many domains, it has been adopted or considered in several key industries, such as air travel, hotel booking,  insurance, and ride-sharing.
For example, according to 	\cite{wiser:18},  ``Hotel websites such as Orbitz (whose parent company is Expedia) and auto dealers like Tesla utilize personalized pricing to their advantage when conducting sales with a customer. Even Uber has dabbled in personalized pricing by offering `premium pricing' to predict which users are willing to pay more to go to a certain location.'' As reported by \cite{Mohammed:17},  when using Orbitz, for identical flights, hotel and type of room, the price of the traveling package found on a laptop was 6.5\% more than the price offered on the Orbitz app.  Moreover, in practice, instead of directly charging different prices, the e-commerce platforms usually use the discount or promotions to implement personalized pricing strategies.

Although the availability of personal data (e.g., location, web search histories, media consumption, social media activities) enables targeted services for an individual customer, it poses significant privacy issues in practice (e.g., \cite{AppleDP:2017}).
Many existing privacy-protection approaches are rather ad hoc by ``anonymizing'' personal information. However, such ad hoc anonymization leads to two issues. First, it is difficult to quantify the level of privacy. Second, it has been shown that a de-anonymization procedure can easily jeopardize privacy. Examples include the de-anonymization of released AOL search logs \citep{aol2006} and movie-watching records in Netflix challenge \citep{Narayanan2008Robust}. Therefore, personalized operations management urgently calls for mathematically rigorous privacy-preserving methods to prevent personal information leakage in online decision-making.  On one hand, personalized revenue management has received a significant amount of attention in recent operations literature (see, e.g., \cite{ban2017personalized,Cheung2017} and references therein). On the other hand, the question of how to protect an individual's privacy has not been well-explored in the existing literature. 

In this paper, we study how to systematically protect an individual's privacy in the dynamic pricing problem with demand learning.  Given $T$ time periods,   a potential customer arrives at each time $t$,  and the retailer receives $x_t$ containing information about the incoming customer,
such as age, location, purchase history and ratings, credit scores, etc.  We consider a very general personalized setting, where the customers are \emph{heterogeneous} and thus the feature  $\{x_t\}_{t=1}^T$ does \emph{not} necessarily follow the same distribution. By observing the personal information $x_t$, the retailer offers the customer a price $p_t\in[0,1]$. The customer then makes 
$y_t\in\mathbb R$, where the random demand $y_t$ follows a \emph{generalized linear model} of a feature vector $\phi(x_t, p_t) \in \mathbb{R}^d$ (see \eqref{eq:exp-family}) and the retailer collects revenue $p_ty_t$. The objective of the retailer is to maximize the expected revenue over the entire $T$ time periods,
or more specifically $\mathbb E[\sum_{t=1}^Tp_ty_t]$.
As this paper focuses on how to protect an individual's sensitive information, we consider a stylized setting of pricing a \emph{single} product, with \emph{unlimited} inventories available.
 
Due to the personalized nature, the aforementioned pricing procedure involves the use of individuals' sensitive information, such as customers' personal information, {characterized by $x_t$} and their purchase history, designated {by $y_t$} (e.g., whether a purchase was made at time $t$). 
Thanks to secured internet communication channels, the information $(x_t, p_t, y_t)$ at time $t$ is usually securely transmitted, and thus only revealed
to the retailer and the particular customer coming at time $t$. However, although the information at time $t$ is not directly accessible to future customers,  the sensitive information is not completely shielded from outside third-party agents (a.k.a. attackers or adversaries) because of the ripple effects of historical customers' data on future pricing decisions in a data-driven pricing system. Indeed,  a third-party agent who observes his own posted prices \emph{in the future} can 
potentially infer an individual's personal information $x_t$ and her purchase decision $y_t$.  We provide two examples showing how the sensitive data at time $t$ could be potentially breached and why such privacy leakage could incur serious challenges to the integrity of the underlying pricing system.



\paragraph{Leakage of purchase activity $y_t$.} 
For sensitive commodities such as medications, customers' purchasing decisions $\{y_t\}$ must be well protected from the public, as such purchases may potentially reveal purchasers' underlying medical conditions. Some dynamic pricing policies would increase prices facing increased sales volumes for a higher profit.
Such behavior might inadvertently leak information about $y_t$ to a third party via the fluctuation of prices.
For example, a third-party agent might place orders immediately before and after a person of interest
and if he sees a slight spike in his received prices,
he might be able to infer the purchase decision $y_t$ of the person of interest.

\paragraph{Leakage of customers' personal information $x_t$.}
When making the price decision $p_t$ for an arriving customer at time $t$, 
the retailer makes use of the customer's personal information $x_t$.
Some components of $x_t$, such as the customer's age, credit history, and prior purchases, are highly sensitive and should be protected. 
Consider a natural pricing policy that is highly ``local'' to personal information, e.g.,  posting similar prices to future customers with a similar profile to customer $t$. A third-party agent could arrive before and after a person of interest with {guesses of personal information} to detect whether there are noticeable changes in the prices. Then, the agent would be able to infer to some degree about the personal information $x_t$ of the individual of interest.




In summary, it is vital to develop systematic and mathematically rigorous policies that \emph{provably} protect customers' privacy.
As we previously discussed, simple data anonymization lacks a theoretical foundation and can be jeopardized.  On the other hand, the notion of \emph{differential privacy} (DP), which was proposed in the computer science field \citep{dwork2006our,dwork2006calibrating},  has laid a solid foundation for private data analysis and achieved great success in industries. The DP is not only a gold standard notion in academia but also has been widely adopted by companies, such as Apple \citep{AppleDP:2017}, Google \citep{Erlingsson:14}, Microsoft \citep{ding2017collecting}, and the U.S. Census Bureau \citep{Abowd2018Census}. 
The aim of this paper is therefore to build upon the differential privacy notion
to design mathematically rigorous privacy policies with provable utility (regret) guarantees for the dynamic personalized pricing problem.

\subsection{Our contributions}\label{subsec:contributions}

The major contributions of this paper can be summarized as follows:


\paragraph{\bf Near-optimal regret of provably privacy-aware pricing policies.}
Built upon the notion of anticipating differential privacy, we propose a privacy-aware personalized pricing algorithm
that enjoys rigorous regret guarantees.
More specifically, in a general setting when the personalized information of each coming customer can be adversarially chosen, 
our policy achieves a regret upper bound of $\tilde{O}(\varepsilon^{-1} \sqrt{d^3T})$, where $\varepsilon$ is the parameter in DP (a smaller $\varepsilon$ implies a stronger privacy preservation of the resulting algorithm), 
$d$ is the dimension of the feature map $\phi(x_t, p_t)$, $T$ is the time horizon, and $\tilde{O}(\cdot)$ hides logarithmic factors (see Theorem \ref{thm:main-regret-adversarial-context}). 
The $\sqrt{T}$ dependency on the time horizon $T$ in this regret upper bound is optimal \citep{broder2012dynamic}.

In addition to the regret upper bound for the general personalized information setting,
we also study a ``stochastic'' setting in which the customer's personal information $\{x_t\}$ is assumed to be stochastic and independently and identically distributed from an unknown non-degenerate distribution.
We remark that this is a common assumption/setting
studied in the existing literature \citep{qiang2016dynamic,Miao2019context}.
In this setting,
with some changes of hyper-parameters of our proposed algorithm,
an improved regret upper bound of $\tilde{O}(d \sqrt{T} + \varepsilon^{-2} d^2 )$ can be proved (see Theorem \ref{thm:main-regret-stochastic-context}).   {One attractive property of this bound is that it separates the dependency on conventional problem parameters (i.e., $d$ and $T$) from privacy-related parameter (i.e., $\varepsilon$).  }
The dominating term (with $T\to\infty$) in this regret bound, namely the $\widetilde O(d\sqrt{T})$ term,
is \emph{optimal} in both $d$ and $T$, as shown in \citep{dani2008stochastic}.

{
In both the general setting and the ``stochastic'' setting, the regret upper bounds of either $\tilde O(\varepsilon^{-1}\sqrt{d^3 T})$
or $\tilde O(d\sqrt{T}+\varepsilon^{-2}d^2)$ also characterize the tradeoffs between customers' privacy protection and the revenue
(surplus) of the seller under the designed policy.
More specifically, the $\varepsilon>0$ parameter characterizes the level of customers' privacy protection, with smaller $\varepsilon$
corresponding to stronger protection against malicious agents.
Clearly, as both regret upper bounds depend inversely on $\varepsilon$, it shows that as the seller seeks stronger protection
over the privacy of customers' personalized data, the more he/she will suffer from decreased revenue (and  a larger regret). This revenue loss is due to additional efforts/randomization required for data privacy protection.
}

Finally, the privacy requirements imposed on the seller's policy also have interesting implications on consumer surplus.
In Sec.~\ref{subsec:discussion-consumer-surplus} of this paper, we provide numerical results to characterize the tradeoffs between consumers' privacy protection and consumer surplus.
We find that as the implied privacy protection becomes weaker (i.e., the seller having less ability to discriminate against customers
based on their personal data and features, resembling a transition from the first-degree to the third-degree price discrimination), 
the consumer surplus increases because the seller extracts less of the consumer surplus from his/her limited ability to carry out price discrimination.

\paragraph{\bf Technical contributions.}
Our proposed framework for privacy-preserving personalized dynamic pricing makes use of several existing privacy-aware
learning/releasing techniques, such as the \textsc{AnalyzeGauss} method in online PCAs \citep{dwork2014analyze},
the tree-based aggregation technique for releasing serial data \citep{chan2011private}, 
and differentially private empirical risk minimization methods \citep{kifer2012private,chaudhuri2011differentially}.
On the other hand, the development and analysis of our proposed method make several key technical contributions
to the general topic of privacy-aware sequential decision-making in revenue management problems, which we briefly summarize as follows:
\begin{enumerate}
\item One salient feature of this paper is the inclusion of customers' personal information $x_t$ as sensitive data that needs to be protected,
which is different from existing works \citep{tang20contextual}, where only purchase activities $y_t$ are regarded as sensitive data (see Section \ref{sec:related} for more discussions).
The objective of protecting privacy in $\{x_t\}$ presents two technical challenges.
First, as $\{x_t\}$ and subsequently the feature representations $\{\phi_t\}$ are sensitive data, one cannot directly apply the private follow-the-regularized-leader (FTRL) approach in \citep{tang20contextual} to the dynamic pricing problem.
Furthermore, the sensitivity of $\{x_t\}$ implies the sensitivity of $\{p_t\}$ as well, since prices offered to incoming customers must be strongly associated with customers' personal information to achieve good revenue performances.  To address these challenges, we build our DP setting on the notion of \emph{anticipating DP} \citep{shariff2018differentially}, which excludes prices in prior selling periods from the outcome sets of a randomized algorithm.
\item The demand rate function $f$ as a function of price $p$ and personal information $x$ is modeled in this paper as a \emph{generalized linear model}
within the exponential family.
Despite its apparent similarity to linear models, such generalization results in significant challenges when privacy concerns
are considered. {In fact, this is still an open problem for generalized linear contextual bandit under the DP guarantee.}
{ More specifically, the results of \cite{shariff2018differentially} on privacy-aware linear bandits rely heavily on the fact that the ordinary least squares solution is in a closed-form with two simple sufficient statistics: the sample covariance matrix $X^\top X$ and the response-weighted feature vector $X^\top y$.
With the post-processing property of DP (which we briefly discuss in Section~\ref{subsec:post-processing}), 
it suffices to obtain privacy-preserved copies of $X^\top X$ and $X^\top y$ at each time.
In contrast, parameter estimates in generalized linear models are usually obtained using maximum likelihood estimates (MLE), which do not have
simple sufficient statistics. It is nearly impossible to guarantee  the privacy and a non-trivial regret simultaneously if the MLE is updated at every period.}
To overcome this challenge, we make the important observation that the required number of updates of MLEs can be reduced significantly {(i.e., only $O(d \log T)$ periods of updates will be sufficient)}.
This key observation allows us to compose differentially private empirical risk minimizers \citep{kifer2012private} to arrive at a privacy-aware
contextual bandit algorithm even without explicit sufficient statistics.
\item The generalized linear model for demand rate modeling resembles existing works on parametric contextual bandits
without privacy constraints \citep{li2017provably,filippi2010parametric,wang2019optimism}.
One significant limitation of these existing works is that, without assuming stochasticity of the contextual vectors, 
the optimization of parameter estimates in these works is usually non-convex.
Examples include the robustified Z-estimation in \citep{filippi2010parametric} and the constrained least-squares formulation in \citep{wang2019optimism},
both of which are non-convex for some popular generalized linear models such as the logistic regression model.
While such non-convexity poses only computational difficulties in non-private bandit algorithms,
these challenges become much more significant when privacy constraints are imposed since most existing techniques of DP
stochastic optimization require convexity \citep{kifer2012private,chaudhuri2011differentially} and the general privacy-aware non-convex optimization is extremely difficult.

To overcome this challenge, this paper analyzes a constrained maximum likelihood estimation in a more refined style
with a relatively large regularization parameter, demonstrating with high probability that the solution to the constrained MLE lies in the strict interior of the constraint set (see Lemma EC.1 in the supplementary material).
This result then implies the first-order KKT condition of the solution, from which the Z-estimation analysis in \citep{li2017provably,filippi2010parametric}
can be used together with the analysis of an objective-perturbed convex minimization problem to obtain satisfactory regret upper bounds.
\end{enumerate}
%
%


\subsection{Organization}
The rest of the paper is organized as follows. Section \ref{sec:related} discusses the related literature in both dynamic pricing and differential privacy. We set up our pricing models and formalize the anticipating DP in Sections \ref{sec:models} and \ref{sec:prelim-privacy}. Our policy is presented in Section \ref{sec:algo}, which contains two components: \emph{privacy releasers} and \emph{price optimizers}. Sections \ref{sec:privacy} and \ref{sec:regret} establish the privacy and regret guarantees, respectively, followed by a conclusion in Section \ref{sec:con}. 
All the technical proofs are relegated to the online supplementary material.

\section{Literature Review}
\label{sec:related}

This section briefly reviews related research from both the personalized pricing and differential privacy literature.

\paragraph{Personalized dynamic pricing with demand learning.} 

Due to the increasing popularity of online retailing, dynamic pricing with demand learning has become an active research area in revenue management in the past ten years (see, e.g., \cite{araman2009dynamic, besbes2009dynamic, farias2010dynamic, harrison2012bayesian, broder2012dynamic, den2013simultaneously, wang2014close, chen2015real, besbes2015suprising, cheung2017dynamic,ferreira2018online, Wang21:uncertainty}).  More recently, due to the availability of abundant personal information, personalized pricing with feature information has been investigated in several works. For example, \cite{chen2020statistical} studied offline personalized pricing and quantified the statistical property of the MLE. \cite{Cohen2020} considered a binary thresholding model for purchasing decisions by comparing a linear function of the feature and the posted price, proposed an ellipsoid-based method for dynamic pricing, and established the worst case regret bound. \cite{qiang2016dynamic} considered a linear demand model and studied the performance of the greedy iterated least squares. \cite{ban2017personalized} and  \cite{javanmard2019dynamic} studied the personalized dynamic pricing problem in high-dimensional settings with sparsity assumption of features. A very recent work by \cite{tang20contextual} studied differentially-private contextual dynamic pricing and proposed a Follow-the-Approximate-Leader-type policy. Our work differs from this paper in several aspects. First, we protect the personal information $\{x_t\}$, while \cite{tang20contextual} treated this information as public. Second, \cite{tang20contextual} adopted the classical DP notion, while we consider the notion of anticipating DP. Finally, we assume that the demand follows a generalized linear model of a feature map of personal information and price, while \cite{tang20contextual} considered a binary thresholding purchase model with a linear mapping of contextual information.

\paragraph{Differential privacy for online learning.} 
Since the notation of $(\varepsilon, \delta)$-differential (DP) privacy was proposed by \cite{dwork2006our, dwork2006calibrating}, it has become a golden standard for privacy-preserving data analysis in both academia and industry. Please refer to the survey \cite{dwork2014algorithmic} for a comprehensive introduction of DP.

Built on this classical notion, other privacy notions have also been developed in the literature, such as Gaussian DP \citep{Dong19Gaussian}, joint DP \citep{shariff2018differentially}, local DP \citep{Evfimievski2003limiting,kasiviswanathan2011can}, average-KL DP \citep{wang16average} and per-instance DP \citep{wang2019per}. Our notion of anticipating DP is motivated by the joint DP \citep{shariff2018differentially} designed for linear contextual bandits.
While the work of \cite{shariff2018differentially}  studied the linear contextual bandits  subject to differential privacy constraints,
their methods and analysis are built upon the noisy perturbation of sufficient statistics (namely, the sample covariance and sample average).
Thus, their method is \emph{not} applicable to the personalized pricing question, where generalized linear demand models are widely used
(see also the technical challenges summarized in the introduction).

In DP, there are several fundamental techniques, such as composition, post-processing (see Section \ref{subsec:comp} and \cite{dwork2014algorithmic}), partial-sum by tree-based aggregation \cite{dwork2010differential,chan2011private}, and ``objective-perturbation'' \citep{chaudhuri2011differentially,kifer2012private}.  
In our designed personalized dynamic pricing algorithm, we build on these important techniques to make sure that our algorithm is differentially private.

The techniques of DP have been applied to multi-armed bandit problems. For example, \cite{mishra2015nearly}  developed differentially private UCB and Thompson sampling algorithms for classical bandits.  \cite{mishra2015nearly} and \cite{shariff2018differentially} further studied  differentially private linear contextual bandits, where \cite{mishra2015nearly}  protected the privacy of rewards and \cite{shariff2018differentially} protected both rewards and contextual information. However, for linear bandits, since the maximum likelihood estimator (MLE) admits a simple closed-form solution, one only needs  to protect the sufficient statistics (e.g., $\sum_{t'=1}^t x_{t'} x_{t'}^\top$ and $\sum_{t'=1}^t y_{t'}x_{t'}$). On the other hand, we consider a much more general demand model following a generalized linear model. Therefore, the corresponding MLE does not admit a closed-form solution; we address this challenge by providing a new analysis of constrained MLE properties.
{There are other interesting private online learning frameworks developed in  recent literature. For example,
the private sequential learning model was proposed in  \cite{Tsitsiklis2020private} (for noiseless responses) and further investigated in \cite{xu2018query}  and  \cite{xu2020optimal} (for noisy responses). In particular, 
\cite{xu2020optimal} quantified the optimal query complexity for private sequential learning against eavesdropping.  While existing privacy literature mainly focuses on protecting a data owner's privacy, this work investigates how to protect the privacy of a learner who sequentially queries a database and receives binary responses.  We note that the goal of the private sequential learning is to learn a global parameter, e.g., ``the highest price to charge so that at least 50\% of the consumers would purchase'' in pricing domain \citep{xu2020optimal}, and to make sure the adversary cannot infer the final released price. In contrast, our goal is to make sequential decision-making to maximize revenue while protecting individuals' personalized information and purchasing decisions.
}

{
In the recent work of \cite{lei2020privacy}, which was completed after this paper was released, an \emph{offline} personalized pricing setting
is studied with differential privacy guarantees.
The recent work of \cite{zheng2020locally} studied the stronger local privacy notion and derived an algorithm with $\widetilde O(T^{3/4})$ regret bound
for the generalized linear model, which is worse than the regret bounds obtained in this paper.
}

\section{Pricing models and assumptions}\label{sec:models}

The basic setting of personalized dynamic pricing has been described in the introduction. In this section, we provide more technical details of the problem setting. At each time $t$ with the observed personal information $x_t$ and the posted price $p_t$, the (random) demand realized by customer at time $t$ is modeled by a Generalized Linear Model (GLM) within the exponential family,  taking the form of
\begin{equation}
\Pr[y_t=y|p_t,x_t,\theta^*] = \exp\{\zeta(y\phi_t^\top\theta^*-m(\phi_t^\top\theta^*)) + h(y)\},
\label{eq:exp-family}
\end{equation}
where $\phi_t=\phi(x_t,p_t)\in\mathbb R^d$ is a known feature map, $\theta^*\in\mathbb R^d$ is an unknown linear model,
and $\zeta,m(\cdot),h(\cdot)$ are components of the distribution family.
Some examples of exponential family distributions include the Gaussian distribution and the Logistic model,
which are given at the end of this section.
It is easy to verify that $f(\phi_t^\top\theta^*):=m'(\phi_t^\top\theta^*)$ is the expectation of $y_t$ conditioned on $p_t,x_t$ and $\theta^*$.
Hence, we can equivalently write Eq.~(\ref{eq:exp-family}) as
\begin{equation}
y_t = f(\phi_t^\top\theta^*) + \xi_t, 
\label{eq:demand-model}
\end{equation}
where $\phi_t=\phi(x_t,p_t)$ and $\xi_t$ are independent random variables satisfying $\mathbb E[\xi_t|p_t,x_t]=0$.

We next specify the filtration process of $x_t$ and $p_t$.
Let $\mathcal F_{t}=\{(x_\tau,y_\tau,p_\tau)\}_{\tau=1}^t$ be the history up to time period $t$.
In the most general setting, the features $\{x_t\}_{t=1}^T$ of the $T$ customers are arbitrarily chosen before the pricing process starts \footnotemark\footnotetext{This setting is known as the ``oblivious adversary'' model in the contextual bandit literature. While this model is weaker than the ``fully adversarial'' one mostly studied in the literature, we adopt the oblivious adversary model for a more convenient treatment of privacy constraints, as $\{x_t\}$ will not depend on the offered prices or the randomly realized demands.}.
The price $p_t$ at each time $t$ is subsequently chosen by the dynamic pricing policy conditioned on filtration $\mathcal F_{t-1}$ and $x_t$.
The demand $y_t$ is then realized via $y_t=f(\phi_t^\top\theta^*)+\xi_t$, where $\phi_t=\phi(x_t,p_t)$ and $\mathbb E[\xi_t|x_t,p_t,\mathcal F_{t-1}]=0$.

Throughout this paper we impose the following conditions on the distribution family, the linear model, and the feature map:
\begin{enumerate}
{
\item There exists a parameter $B_Y<\infty$ such that $|y_t|\leq B_Y$ for all time periods $t$ in all databases $D$;
}
\item Both the feature vectors and the linear model have at most unit norm, or more specifically $\|\phi(x,p)\|_2,\|\theta^*\|_2\leq 1$ for all $x,p$;
\item The stochastic noises $\{\xi_t\}$ are centered and sub-Gaussian, meaning that $\mathbb E[\xi_t|x_t,p_t,\mathcal F_{t-1}]=0$
and there exists $s<\infty$ such that $\mathbb E[e^{\lambda\xi_t}|x_t,p_t,\mathcal F_{t-1}]\leq e^{\lambda^2s^2/2}$ for all $\lambda\in\mathbb R$;
\item $f(\cdot)=m'(\cdot)$ maps $\mathbb R$ to $[0,1]$ is continuously differentiable and strictly monotonically increasing.
Furthermore, for all $|z|\leq 2$, $K^{-1}\leq f'(z)\leq K$ for some constant $1\leq K<\infty$;
\item $\zeta$ in Eq.~(\ref{eq:exp-family}) satisfies $G^{-1}\leq \zeta\leq G$ for some constant $1\leq G<\infty$.
\end{enumerate}

We give some common examples that fall into Eq.~(\ref{eq:exp-family}) and satisfy all imposed conditions.

\begin{example}[Gaussian model]
In the Gaussian model the realized demand $y_t$ follows $y_t=\phi_t^\top\theta^*+\xi_t$ with $\xi_t\sim\mathcal N(0,1)$.
It is easy to verify that the Gaussian model falls into Eq.~(\ref{eq:demand-model}) with $\zeta=1$,
$m(z)=\frac{1}{2}z^2$, $f(z)=m'(z)=z$, and $h(y)=-\frac{1}{2}y^2-\frac{1}{2}\ln(2\pi)$.
The Gaussian model also satisfies all imposed conditions with high probability with $B_Y\lesssim s\sqrt{\ln T}$, $s=1$, $K=1$, and $G=1$.
\end{example}
\begin{example}[Logistic model]
In the Logistic model the realized demand $y_t$ is supported on $\{0,1\}$,
following the Logistic distribution $\Pr[y_t=1|\phi_t,\theta^*] = e^{\phi_t^\top\theta^*}/(1+e^{\phi_t^\top\theta^*})$.
It is easy to verify that the Logistic model falls into Eq.~(\ref{eq:demand-model}) with $\zeta=1$,
$m(z)=\ln(1+e^z)$, $f(z)=m'(z)=e^z/(1+e^z)$, and $h(y)=1$.
The Logistic model also satisfies all imposed conditions with $B_Y=1$, $s=1$, $K=(1+e^2)^2/e^2$, and $G=1$.
\end{example}


\section{Preliminaries on differential privacy}\label{sec:prelim-privacy}

In this section we present background material on \emph{differential privacy}, the core privacy concept
adopted in this paper.
We start with the introduction of the standard differential privacy concept, and then show
how the privacy concept could be extended to its ``anticipating'' version which is more appropriate for data-driven sequential decision-making problems.
Finally we discuss two fundamental concepts of
\emph{composition} and \emph{post-processing}, which are essential in designing complex differentially private systems.
For a full technical treatment and historical motivations, the readers are referred to the comprehensive review by \cite{dwork2014algorithmic}.

\subsection{Differential privacy}

\emph{Differential privacy} is a mathematically rigorous measure of privacy protection and has been extensively studied and applied
since its proposal in the work of \cite{dwork2006calibrating}.
At a higher level, the fundamental concept behind differential privacy is the \emph{impossibility} of distinguishing two ``neighboring databases'' (differing only on a single entry) 
with high probability, based on publicly available information about the database.
To facilitate such probabilistic indistinguishability, the conventional approach is to artificially calibrate \emph{stochastic noise} into the process
or the outputs of differentially private algorithms.

\begin{figure}[t]
\centering
\includegraphics[width=0.7\textwidth]{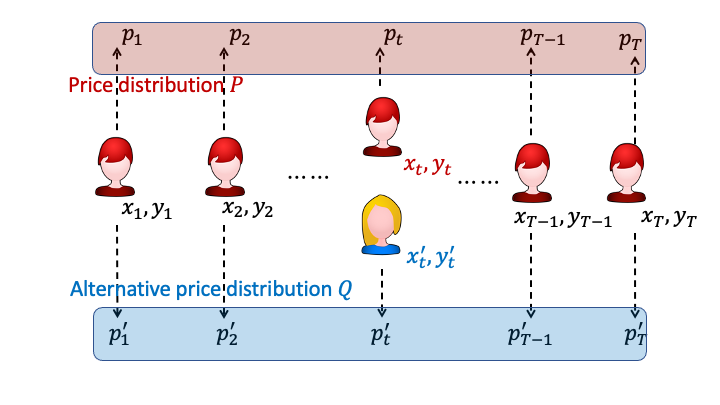}
\caption{{Illustration of the differential privacy concept.}}
\label{fig:dp-illustration}
\end{figure}

{

More specifically, Figure \ref{fig:dp-illustration} gives an intuitive illustration of the differential privacy concept
applied to our dynamic personalized pricing problem.
Suppose at time $t$ the incoming customer with the context vector $x_t$ is being offered price $p_t$ and makes purchase decision $y_t$.
The price decisions $\{p_t\}_{t=1}^T$ produced by the pricing algorithm are usually random, and therefore we can use $P$ to denote
the joint distribution of these prices.
The concept of differential privacy requires that, if a customer's personal data change from $(x_t,y_t)$ to $(x_t',y_t')$,
while all the other $T-1$ customers' data remain unchanged, 
the joint distributions of the posted prices $P$ will change to a distribution $Q$ that is very close to $P$.
The closer $P$ and $Q$ are under the hypothetical personal data change $(x_t,y_t)\to (x_t',y_t')$, 
the better data privacy is protected under the pricing policy.

Why is the close proximity of price distributions $P$ and $Q$ a good measurement of a pricing algorithm's privacy protection?
Assume that a malicious agent  would like to extract the sensitive information of a particular customer of interest, who arrives in the system
at time $t$. The malicious agent must extract such sensitive data based solely on \emph{publicly available} information,
which in this case would be the firm's posted prices $p_1,\cdots,p_T$. {Here, ``public information'' in the differential privacy literature refers to the information or released data that can be accessed by a malicious adversary, because these data are used by the adversary to infer the personalized data of the customers, whose privacy is to be protected.} 
If the price distributions $P$ and $Q$ produced by the pricing algorithm are very similar,
then it is \emph{information-theoretically} not possible for the malicious agent to distinguish with reasonable success probability between a customer $(x_t,y_t)$ and another hypothetical customer $(x_t',y_t')$ (see Figure \ref{fig:dp-illustration}).
This means that no matter how smart the malicious agent is, it is impossible for him to extract very much sensitive data
from the customer of interest simply based on publicly available price information.

Mathematically, we use $D$ to denote the database of all sensitive data $\{(x_t,y_t)\}_{t=1}^T$ for all of the $T$ customers.
For convenience of presentation, we also write $o_t=(x_t,y_t)$.
A database $D'$ that is a \emph{neighboring database} of $D$ if and only if $D'$ and $D$ only differ at a single time period.
More specifically, $D=\{o_t\}_{t=1}^T,D'=\{o_t'\}_{t=1}^T$ are neighboring databases if there exists $t$ such that $o_t\neq o_t'$
and $o_\tau=o_\tau'$ for all $\tau\neq t$.
Suppose a pricing algorithm $A$ operates with input database $D$ and produces randomized price output $A(D)=(p_1,\cdots,p_T)$. 
The following definition gives a rigorous formulation of $(\varepsilon,\delta)$-differential privacy:
}

\begin{definition}[$(\varepsilon,\delta)$-differential privacy \citep{dwork2006our}]
For $\varepsilon,\delta>0$, a randomized algorithm $A$ satisfies $(\varepsilon,\delta)$-differential privacy if for every pair of neighboring databases $D,D'$
and measurable set $\mathcal A\subseteq [\underline p,\overline p]^T$, it holds that
$$
\Pr[A(D)\in\mathcal A]\leq e^{\varepsilon}\Pr[A(D')\in\mathcal A] + \delta.
$$
\label{defn:eps-delta-privacy}
\end{definition}

{ To facilitate the understanding of this definition, we explain why the multiplicative factor $e^{\varepsilon}$ is critical and the role of the parameter $\delta$ in practice.
Let us first explain why the DP-definition in Def.~\ref{defn:eps-delta-privacy} adopts a multiplicative factor $e^{\varepsilon}$ rather than an additive bound of  $|\Pr[A(D)\in\mathcal A]- \Pr[A(D')\in\mathcal A]$.	Imagine two neighboring datasets $D,D'$ give rise to the same output $O$ with probabilities $p_1 = \Pr[O|D]$ and $p_2=\Pr[O|D']$. The key is to prevent a malicious party from distinguishing between $D$ and $D'$ based on the observation of $O$. If an additive guarantee is involved $|p_1-p_2|\leq\varepsilon$, then it is possible that $p_1=0$ and $p_2=\varepsilon$. If this is the case, the adversary would be \emph{100\% sure whether the underlying dataset is $D$ or $D'$} once she observes the output $O$ (since $p_1=0$ implies that it is \emph{impossible} to observe $O$ given $D$). This means that with probability $\varepsilon$, which is usually not that small (e.g., $\varepsilon=0.1$), a \emph{catastrophe} (i.e., an outside adversary being \emph{completely certain} about the customer's private data) will occur with 10\% probability. On the other hand, if the guarantee is multiplicative (e.g., $0.9p_2\leq p_1\leq 1.1p_2$) then the adversary \emph{cannot}
completely distinguish between $D$ and $D'$ no matter how small $p_1$ or $p_2$ is. Following this discussion on the multiplicative factor versus the additive factor, since $\delta$ is an additive term, it corresponds to the probability of a \emph{catastrophe} happening that allows the adversary to completely infer the privacy information about customers' data.  Since we don't want a catastrophe to happen, $\delta$ needs to be set \emph{overwhelmingly} small.  With a tiny $\delta$ value in the DP-definition, more specifically, the adversary is \emph{always} able to conclude that $D$ (or $D'$) is more likely than the other,
but such preference of likelihood is never going to exceed a ratio of $e^{\varepsilon}$. For example, with $\varepsilon=0.1$, the adversary may conclude that $D$ is 10.5\% more possible than $D'$ based on his observations of published data $O$, but will never be able to completely/deterministically distinguish $D$ from $D'$ based on $O$.

}


\subsection{Anticipating differential privacy}
\label{sec:adp}

Despite being a widely adopted measure, the DP notion as stated in Definition \ref{defn:eps-delta-privacy}
cannot be directly applied to dynamic pricing for several reasons. First, Definition \ref{defn:eps-delta-privacy}
would not lead to useful pricing policies.  
This is because, essentially, 
Definition \ref{defn:eps-delta-privacy} requires that conditioned on the output of the \emph{entire}
posted price sequence, the adversary cannot distinguish between $o_t$ and $o_t'$ in a probabilistic sense. 
On the other hand, for high-profit personalized pricing policies, once the customer's personal information $x_t$ changes, the price $p_t$ offered to that customer must change accordingly 
in order to achieve high expected revenue, making inference of $x_t$ much easier given $p_t$.
Furthermore, as we have discussed in the previous paragraphs, the communications of $(x_t,p_t,y_t)$ at time $t$ are secured in practice and
therefore, an adversary should not have the capability of accessing the price $p_t$ at time $t$. From this perspective,  the classical DP notion defined in Definition \ref{defn:eps-delta-privacy}
is too strong since it implicitly allows the adversary to access the price at time $t$ (as $p_t$ belongs to the output $A(D)$).
In a practical setting, however, the adversary is only able to access information during other time periods (e.g., by maliciously sending fake customers to obtain price quotes) to infer the sensitive information about an individual at time $t$. {In other words, in the following anticipating DP definition (see Definition \ref{defn:na-eps-delta-privacy}), the price offered to a specific customer of interest  $p_t$ is not public information, as we can expect basic communication security between the customer and the seller. However, prices offered to \emph{other} customers are considered public information because a malicious adversary could pretend to be a customer and extract such price information, and subsequently infer the private data of the customer of interest based on such extracted price information.}

This argument can be made rigorous by the following proposition.
{ The proposition is similar to Claim 13 of \cite{shariff2018differentially}, by showing that \emph{any} policy satisfying the $(\varepsilon,\delta)$-differential privacy
in Definition \ref{defn:eps-delta-privacy} must suffer regret that is linear in the time horizon $T$.
The proof of Proposition \ref{prop:anticipating-linear-regret} is, however, different from \cite{shariff2018differentially}, since we study
generalized linear models such as the logistic regression model.
We relegate the complete proof to the supplementary material.}
\begin{proposition}
Let $\pi$ be a contextual pricing policy over $T$ periods that satisfies $(\varepsilon,\delta)$-differential privacy as defined in Definition \ref{defn:eps-delta-privacy},
with $\varepsilon<\ln(2)$ and $\delta<1/4$.
Then the worst case regret of $\pi$ is lower bounded by $\Omega(T)$.
\label{prop:anticipating-linear-regret}
\end{proposition}

\begin{figure}[t]
\centering
\includegraphics[width=0.7\textwidth]{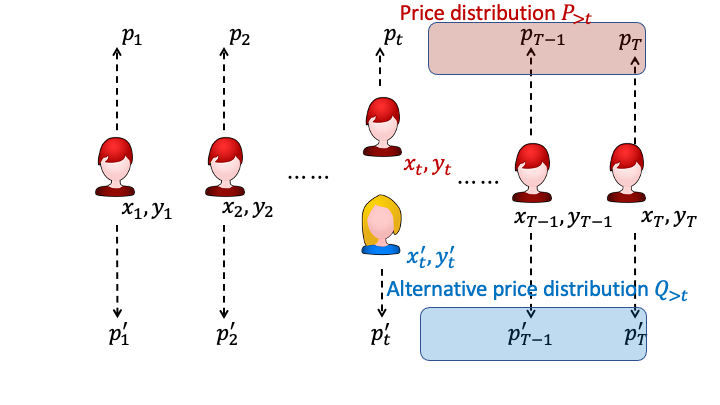}
\caption{Illustration of the anticipating differential privacy (ADP) concept.}
\label{fig:adp-illustration}
\end{figure}

To address the challenges mentioned, \cite{shariff2018differentially}  proposed a notion of  ``joint DP'' in the context of linear contextual bandits. We adopt this notion but refer to it as \emph{anticipating DP}. The  notion of anticipating DP highlights the key property of this definition and our focus on more general dynamic personalized pricing policies. 
{
Figure \ref{fig:adp-illustration} gives an illustration of the anticipating differential privacy (ADP) concept.
Compared to the classical differential privacy notion illustrated in Figure \ref{fig:dp-illustration},
the important difference of ADP is to restrict the output sets to prices strictly \emph{after} a customer of interest $t$
and to only require the distributions of \emph{anticipating} prices (denoted by $P_{>t}$ and $Q_{>t}$) to remain stable
with change of personal information $(x_t,y_t)\to (x_t',y_t')$ at time $t$.
Such a restriction is motivated by the fact that the communication about $(x_t, p_t, y_t)$ at time $t$ is secured and the data prior to time $t$ has no impact on 
the privacy of customer $t$ since the pricing algorithm has no knowledge of $x_t$ before time $t$.
With the formulation of anticipating differential privacy, the challenges we mentioned earlier are resolved
because the pricing decision $p_t$ at time $t$ is no longer in the information set of a potential attacker.
}

Our next definition gives a rigorous mathematical formulation of the anticipating differential privacy notion
illustrated in Figure \ref{fig:adp-illustration}.

\begin{definition}[anticipating $(\varepsilon,\delta)$-differential privacy]
Let $\varepsilon,\delta>0$ be privacy parameters.
A dynamic personalized pricing policy $\pi$ satisfies anticipating $(\varepsilon,\delta)$-differential privacy
if for any pair of neighboring databases $D,D'$ differing at time $t$ (i.e., $o_t\neq o_t'$) and measurable set $\mathcal P_{>t}$, it holds that
\begin{equation}
\Pr[p_{t+1},\cdots,p_T\in\mathcal P_{>t}|\pi,D] \leq e^\varepsilon \Pr[p_{t+1},\cdots,p_T\in\mathcal P_{>t}|\pi,D'] + \delta.
\label{eq:non-ant-private}
\end{equation}
\label{defn:na-eps-delta-privacy}
\end{definition}


We also remark that all privacy definitions in this section are \emph{model-free}, meaning that they do \emph{not} depend 
on how realized demands $y_t$ are modeled.
Hence, the privacy guarantees of our proposed algorithm are independent from the generalized linear demand model in Eqs.~(\ref{eq:exp-family}, \ref{eq:demand-model}).
This fact is essential in practical implementations of privacy-aware algorithms because one cannot build privacy guarantees of an algorithm
on a specific underlying model, which may or may not hold in reality.
The modeling assumptions, on the other hand, are required for \emph{performance analysis} (also known as \emph{utility analysis}, e.g., regret upper bounds or convergence results) of
our proposed privacy-aware pricing policies.

\subsection{Composition in differential privacy} \label{subsec:comp}

When a differentially private algorithm only outputs a single statistic (e.g., the sample mean of the database), Definition \ref{defn:eps-delta-privacy}
is easy to check and verify.
{ In reality, however, a useful differentially private protocol is tasked to release several statistics (sometimes with adaptively chosen queries)
and the \emph{entire} output sequence of a protocol needs to be differentially private.
With multiple output statistics, Definition \ref{defn:eps-delta-privacy} involves high-dimensional vector spaces and 
is therefore difficult to check and verify.
\emph{Composition}, on the other hand, provides convenient \emph{upper bounds} on the privacy guarantee of composite outputs using privacy guarantees of individual queries. Take the dynamic pricing setting as an example. The seller repeatedly interacts with the potential customers by offering different prices. It is therefore essential to leverage a composition guarantee in Fact 1 to make sure that all the prices offered, \emph{when aggregated as a whole}, do not leak consumers' privacy via their personalized data. 
}

{
The left panel of Figure \ref{fig:composition-post-processing} gives an illustration of the concept of composition in the context of personalized pricing.
In this simple example, a centralized pricing algorithm has access to a pool of past customers' sensitive data
and offers personalized prices to three customers.
The rule of composition in differential privacy asserts that the privacy guarantee of the pricing algorithm \emph{worsens} as the pricing algorithm 
offers prices to more customers, each time with access and calculations based on the majority of the same sensitive data.
In particular, if the privacy guarantee for each individual pricing decision is $\varepsilon$, then the joint privacy guarantee 
when $k$ individualized prices are offered will worsen to $\Omega(k\varepsilon)$ or $\Omega(\sqrt{k}\varepsilon)$, depending on the 
detailed composition mechanisms.
}

\begin{figure}[t]
\centering
\includegraphics[width=0.48\textwidth]{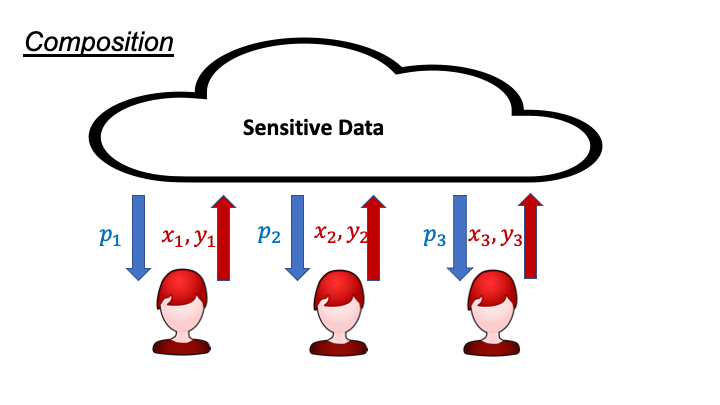}
\includegraphics[width=0.48\textwidth]{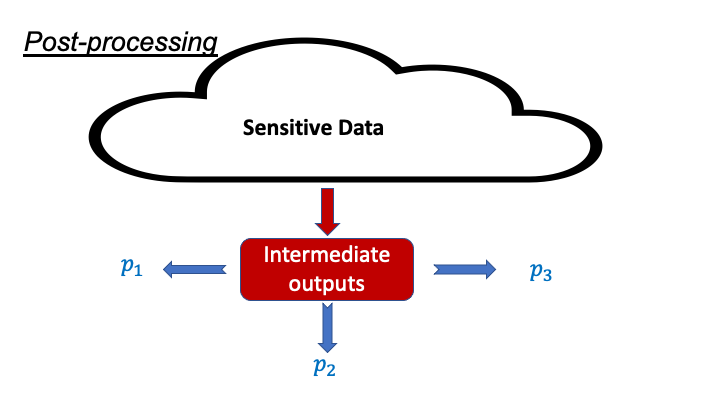}
\caption{Illustration of the concepts of {composition} (left) and {post-processing} (right) in differential privacy.}
\label{fig:composition-post-processing}
\end{figure}

More specifically, let $A=(A_1,\cdots,A_k)$ be a collection of $k$ adaptively chosen queries and suppose that each query $A_k$ satisfies
$(\varepsilon,\delta)$-differential privacy as defined in Definition \ref{defn:eps-delta-privacy}.
The following result is standard in the literature and cited from Theorems 3.16 and 3.20 from \citep{dwork2014algorithmic}.
\begin{fact}
The composite query $A=(A_1,\cdots,A_k)$  satisfies $(\varepsilon',\delta')$-differential privacy with either one of the following:
\begin{enumerate}
\item \emph{(Basic composition)} $\varepsilon' = k\varepsilon$, $\delta'=k\delta$;
\item \emph{(Advanced composition)} $\varepsilon'=\sqrt{2k\ln(1/\tilde\delta)}\varepsilon + k\varepsilon(e^{\varepsilon}-1)$, $\delta'=k\delta+\tilde\delta$
for $\tilde\delta>0$.
\end{enumerate}
\label{fact:composition}
\end{fact}

{
To avoid potential confusion, we remark that both basic and advanced composition apply to \emph{any} differentially private algorithms.
Indeed, they are two different types of joint privacy guarantees proved using different techniques, reflecting different tradeoffs
when composing multiple differentially private queries/algorithms together.
In particular, the basic composition shows a linear growth in the $\varepsilon$ parameter (i.e., $\varepsilon'=k\varepsilon)$
but it allows the $\delta'$ parameter to be zero when the individual queries are $(\varepsilon,0)$-private.
On the other hand,  the advanced composition allows for a slower growth of the $\varepsilon$ parameter (i.e., $\varepsilon'\asymp \sqrt{k}\varepsilon$)
but must yield an $(\varepsilon',\delta')$ differential privacy guarantee with $\delta'>0$, even if the individual queries are $(\varepsilon,0)$-private.
In this paper, we shall use primarily the advanced composition result because we focus on $(\varepsilon,\delta)$ privacy guarantees with $\delta>0$.
}

\begin{corollary}[Corollary 3.21, \citep{dwork2014algorithmic}]
Given target privacy level $0<\varepsilon'<1$, $\delta'>0$ of the composite query $A$,
it is sufficient for each sub-query to be $(\varepsilon,\delta)$-differentially private with $\varepsilon=\varepsilon'/2\sqrt{2k\ln(2k/\delta)}$ and $\delta = \delta'/2k$.
\label{cor:advanced-composition}
\end{corollary}

\subsection{Post-processing in differential privacy}\label{subsec:post-processing}

Practical privacy-aware algorithms usually involve several separate sub-routines.
In most of the cases, not all sub-routines access the sensitive database: some sub-routines may only process
the results from other sub-routines.
{ The principle of \emph{post-processing} states that one only needs to preserve the privacy of those sub-routines with access to the sensitive database
in order to argue for privacy protection of the entire algorithm. For example, in dynamic pricing, algorithms are developed into different components and only one of them directly accesses the 
sensitive data. It is therefore necessary to use the concept of post-processing to argue that the entire algorithm viewed as a whole
does not leak consumers' private personalized data.}

{
The right panel of Figure \ref{fig:composition-post-processing} gives an intuitive illustration of the post-processing concept in differential privacy.
Suppose that an algorithm with full access to all sensitive data has produced some intermediate results (as shown in the red square of the illustration), and these intermediate results have already satisfied the definitions of differential privacy.
Further assume that there is a downstream algorithm, which operates arbitrarily on the intermediate results to produce the personalized prices $p_1,p_2,p_3,\ldots$, \emph{without accessing the sensitive data any more}. 
Then the \emph{post-processing} asserts that there is no need to worry about potential privacy leakages of the downstream algorithm because the intermediate results have already been privatized. 
This useful concept makes it easier to design multi-step, sophisticated privacy-preserving algorithms.
}

More specifically, let $A$ be a sub-routine with access to the sensitive database and $B$ be a sub-routine that only depends on the results of $A$.
\begin{fact}[Proposition 2.1, \citep{dwork2014algorithmic}]
Suppose the outputs of sub-routine $A$ satisfy $(\varepsilon,\delta)$-differential privacy.
Then the outputs of sub-routine $B$ also satisfy $(\varepsilon,\delta)$-differential privacy.
\label{fact:post-processing}
\end{fact}

\section{Algorithmic framework}
\label{sec:algo}
In this section we present the framework of our proposed privacy-aware dynamic personalized pricing algorithm.

{ A straightforward idea is to directly inject noise into customers' sensitive information (e.g., $x_t$) to protect privacy. However, as we will explain later in the paper (see Section~\ref{subsec:input-perturbation}), such a method will fail because the features of each individual customer are relatively independent of each other. Thus, an excessively large magnitude of noise needs to be injected, which incurs a large regret. Therefore, this paper will develop a new dynamic personalized pricing algorithm based on the privacy-preserving maximum likelihood estimator.
}
To better illustrate our algorithm, we first introduce two types of routines used in our algorithm: the \emph{private releasers} that access the sensitive database
and produce differentially private outputs, and the \emph{price optimizers} that access only the outputs from private releasers
to assign near-optimal and privacy-aware prices.
Then a pseudo-code description of our main algorithm will be presented and discussed.

\subsection{Private releasers and price optimizers}

Our proposed privacy-preserving dynamic personalized pricing algorithm consists of several sub-routines.
We divide the sub-routines into two classes: the \emph{private releasers} and the \emph{price optimizers}.

The \emph{private releasers} access the sensitive database $\{x_t,p_t,y_t\}_{t=1}^T$ and output differentially private intermediate results.
For example, in Figure \ref{fig:framework} the \textsc{PrivateCov} routine returns differentially private sample covariance matrices 
and the \textsc{PrivateMLE} routine returns differentially private maximum likelihood estimates.
For private releaser routines, the differential privacy notions are classical (in Definition \ref{defn:eps-delta-privacy}).
Note that, in addition to differential privacy guarantees, the sub-routines also need to satisfy the anticipating constraints for pricing algorithms (i.e.,
accessing only $\{x_\tau,y_\tau,p_\tau\}_{\tau<t}$ to produce any outputs being used at time $t$).

The \emph{price optimizer}, on the other hand, performs optimization and outputs the prices $p_t$ for each time period $t$.
To ensure privacy, our designed price optimizer will \emph{not} directly access historical sensitive data $\{x_\tau,y_\tau,p_\tau\}_{\tau<t}$.
Instead, it optimizes the offering price $p_t$ based only on $x_t$ (the personal information of the incoming customer) and intermediate quantities computed
by private releasers up to time $t$.

Because our designed price optimizer has access to $x_t$ at time $t$, one cannot directly apply the post-processing rule in Fact \ref{fact:post-processing}
to argue privacy guarantees. Nevertheless, the following proposition shows that if all private releasers are differentially private, then so is the price optimizer
in the sense of anticipating differential privacy in Definition \ref{defn:na-eps-delta-privacy}.
The proof of Proposition \ref{prop:na-privacy} is placed in the supplementary material.
\begin{proposition}
Let $(a_1,\cdots,a_T)$ be the outputs of private releasers at each time period $t$ and suppose the entire output sequence $(a_1,\cdots,a_T)$
satisfies $(\varepsilon,\delta)$-differential privacy.
Suppose the price $p_t$ at time $t$ is a deterministic function of $x_t$ and $a_1,\cdots,a_{t-1}$.
Then the pricing policy satisfies anticipating $(\varepsilon,\delta)$-differential privacy.
\label{prop:na-privacy}
\end{proposition}
\begin{remark}
The conclusion in Proposition \ref{prop:na-privacy} holds for $p_t$ as randomized functions of $x_t$, $a_1,\cdots,a_{t-1}$ as well.
Nevertheless, because in our proposed algorithm the price optimizer is deterministic, we shall restrict ourselves to deterministic functions.
\end{remark}

\subsection{Our policy}\label{subsec:policy}

\begin{figure}[t]
\centering
\includegraphics[width=0.9\textwidth]{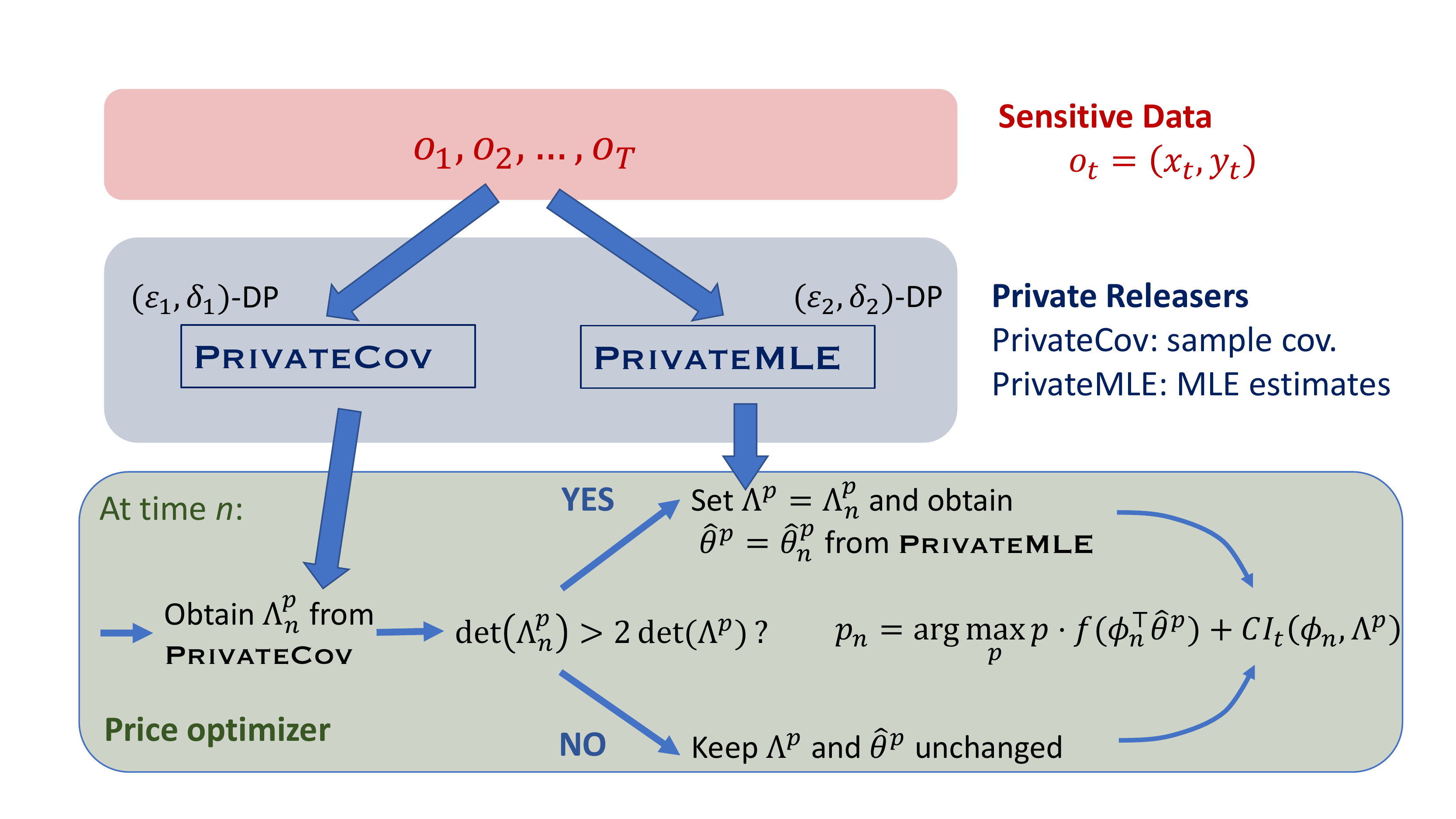}
\caption{Our algorithm framework. Details and explanations in Section~\ref{subsec:policy} in the main text.}
\label{fig:framework}
\end{figure}

In Figure \ref{fig:framework} we depict a high-level framework of our privacy-aware dynamic personalized pricing policy.
It shows a three-layer structure of the proposed policy.
The first layer is the \emph{sensitive database}, consisting of data $\{o_t=(p_t,x_t,y_t)\}_{t=1}^T$; and its privacy needs to be protected.
The second layer is \emph{private releasers}, which consists of two sub-routines \textsc{PrivateCov} (see Algorithm \ref{alg:private-cov} in Section \ref{subsec:privacy-cov}) and \textsc{PrivateMLE} (see Algorithm \ref{alg:private-mle} in Section \ref{subsec:private-mle}).
The \textsc{PrivateCov} sub-routine supplies differentially private sample covariance matrices $\Lambda_n^p\in\mathbb R^{d\times d}$
at every time period.
The \textsc{PrivateMLE} sub-routine outputs differentially private maximum likelihood estimates $\hat\theta_n^p$, but only when such estimates
are requested by the price optimizer.
The \textsc{PrivateCov} sub-routine is designed to be $(\varepsilon_1,\delta_1)$-differentially private and the \textsc{PrivateMLE} routine is 	$(\varepsilon_2,\delta_2)$-differentially private,
so that all outputs from private releasers are $(\varepsilon_1+\varepsilon_2,\delta_1+\delta_2)$-differentially private, thanks to the basic composition rule 
in Fact \ref{fact:composition}.

The third layer of our proposed policy is the price optimizer.
As discussed in the previous section, to ensure privacy the price optimizer shall not access the sensitive database $D$ directly.
Instead it should base its decision of $p_t$ on outputs from private releasers and $x_t$ only.
The last block in Figure \ref{fig:framework} illustrates the basic flow of our price optimizer.
The price optimizer maintains $\Lambda^p$ and $\hat\theta^p$ throughout the pricing process, both of which are obtained directly
from private releasers without accessing the sensitive database.
At the beginning of time period $n$, the price optimizer first obtains sample covariance $\Lambda_n^p$ from the \textsc{PrivateCov} routine.
The optimizer then decides whether to request fresh MLE from the \textsc{PrivateMLE} routine by comparing $\det(\Lambda_n^p)$ with $\det(\Lambda^p)$,
in addition to some other criteria specified in Algorithm \ref{alg:framework}.
Afterwards, $p_t$ is selected as the maximizer of an upper confidence bound of the expected revenue on $x_t$.
It is only during this step that the personal information $x_t$ is involved.

\begin{algorithm}[t]
\caption{The framework of privacy-aware dynamic personalized pricing}
\label{alg:framework}
\begin{algorithmic}[1]
\State \textbf{Input}: privacy parameters $\varepsilon_1,\delta_1,\varepsilon_2,\delta_2>0$, number of pure-exploration periods $T_0$,
maximum number of \textsc{PrivateMLE} calls $D_{\infty}$, regularization parameter $\rho\geq 1$, confidence parameter $\gamma>0$.
\State \textbf{Output}: the offering prices $p_1,p_2,\cdots,p_T$;
\State $\delta_2'=\frac{\delta_2}{2D_\infty}$, $\varepsilon_2'\gets\frac{\varepsilon_2}{2\sqrt{2D_\infty\ln(1/\delta_2')}}$, 
$\Lambda^p = \rho I_d$, $\hat\theta^p=0$, $D_{\mathrm{MLE}}=0$;
\State For the first $T_0$ time periods, offer prices $p_t$ uniformly at random from $[0,1]$;
\For{$n=T_0+1,\cdots, T$}
	\State Obtain $\Sigma_n^p\gets \textsc{PrivateCov}(n,\varepsilon_1,\delta_1)$ and let $\Lambda_n^p=\Sigma_n^p+\rho I_d$;
	\If{$\det(\Lambda_n^p)> 2\det(\Lambda^p)$ and $D_{\mathrm{MLE}} < D_\infty$}
		\State $\hat\theta^p\gets \textsc{PrivateMLE}(n,\rho,\varepsilon_2',\delta_2')$, $\Lambda^p\gets\Lambda_n^p$, $D_{\mathrm{MLE}}\gets D_{\mathrm{MLE}}+1$;
	\EndIf
	\State Offer price $p_n = \arg\max_{p\in[0,1]}\min\{1, pf(\phi_n^\top\hat\theta^p) + \gamma\sqrt{\phi_n^\top(\Lambda^p)^{-1}\phi_n}\}$,
	where $\phi_n=\phi(x_n,p_n)$;
\EndFor
\end{algorithmic}
\end{algorithm}

Algorithm \ref{alg:framework} also gives a pseudo-code description of our proposed pricing policy, which is more accurate and detailed than Figure \ref{fig:framework}.
Note that Algorithm \ref{alg:framework} involves several algorithmic parameters, such as $T_0,D_\infty,\gamma$, and $\rho$,
which do \emph{not} affect the privacy guarantees of the algorithm but do have an impact on its performance.
How to set these algorithmic parameters will be given later in Section~\ref{sec:regret} when we analyze the regret performance of Algorithm \ref{alg:framework}. Before that, we will first make a few important remarks about  Algorithm \ref{alg:framework}.

\begin{remark}[Time complexity] The time complexity for the \textsc{PrivateMLE} sub-routine is the same as traditional maximum likelihood estimation calculations,
if not easier (since the overall formulation is convex), because only the objective is perturbed with a linear term.
The time complexity for the \textsc{PrivateCov} sub-routine is slightly more expensive: at each time $n$, the tree-based protocol needs to update $O(\log n)$ nodes on the binary tree instead of just adding $\phi_t\phi_t^\top$ to a counting matrix.
Overall, the algorithm's time complexity is $O(d^3 T\ln T)$ (note $d^3$ comes from the computation of the determinant), in addition to $O(d\ln T)$ number of MLE calculations.
The next section gives more details on the two private releasers.
\end{remark}

{ 
\begin{remark}[Difference from Generalized Linear Contextual Bandit]
This remark explains how the algorithm differs from a classic generalized linear bandit algorithm without privacy consideration.
The major difference is that when there is no privacy consideration, there is no need (and no use) to randomize 
and therefore vanilla maximum likelihood estimation (MLE) can be used to obtain an estimated model $\widehat\theta_t$ at every time period $t$,
with standard statistical analysis of the errors for such estimates (see \cite{li2017provably}). 
With privacy constraints, such maximum likelihood estimates need to be carefully privatized by calibrating artificial noise into the objective
of the MLE (the \textsc{PrivateMLE} sub-routine later in Algorithm \ref{alg:private-mle}), which also calls for more detailed perturbation-based statistical analysis.
Another difference is that without privacy constraints, the seller could update its model estimate $\widehat\theta_t$ at \emph{every} time period
to obtain the most accurate and updated information.
With privacy constraints, however, the seller cannot afford to adaptively compute a model estimate after each time period due to 
composition constraints and must perform such model estimates sparingly, relying further on a signal scheme also privatized
by incorporating artificial noise matrices (see the \textsc{PrivateCov} sub-routine later in Algorithm \ref{alg:private-cov}).
\end{remark}
}

{ 
\begin{remark}[Exploration Phase]
In addition, we also clarify that the forced exploration step in our algorithm is optional: the proposed algorithm remains valid 
(i.e., satisfying suitable differential privacy constraints and achieving small overall regret) without the forced exploration step (see Theorem \ref{thm:main-regret-adversarial-context} in Sec.~\ref{sec:general_case} where $T_0=0$). 
The forced exploration helps to ensure \emph{improved} regret guarantee when  there are additional distributional assumptions  on  contextual vectors (see Section \ref{subsec:improved}).
This forced exploration aims to make sure the sample covariance of the context vectors is well-conditioned, which leads to improved regret guarantees of privatized MLE. 
\end{remark}
}

\section{Design and analysis of private releasers}\label{sec:privacy}

In this section, we give detailed designs of the two private releasers: the \textsc{PrivateCov} sub-routine
and the \textsc{PrivateMLE} sub-routine.
We prove that both of them satisfy $(\varepsilon,\delta)$-differential privacy as defined in Definition \ref{defn:eps-delta-privacy}.
We also prove several utility guarantees that will be helpful later in the regret analysis of the pricing policy.
{ Figure \ref{fig:proof-flow} shows the flow of our proof framework.}
Due to space constraints and exposition concerns, all proofs to technical lemmas or propositions in this section are placed in the supplementary material.

\subsection{The \textsc{Priva	teCov} sub-routine}\label{subsec:private-cov}
\label{subsec:privacy-cov}
Algorithm \ref{alg:private-cov} gives a pseudo-code description of the \textsc{PrivateCov} sub-routine.
Note that in Algorithm \ref{alg:private-cov} the $\Sigma_n^p$ covariance matrices are released sequentially once each time period,
and $\textsc{PrivateCov}(n,\varepsilon,\delta)$ would simply be the $\Sigma_n^p$ matrix released at the end of iteration $n-1$.

\begin{algorithm}[t]
\caption{The \textsc{PrivateCov} sub-routine}
\label{alg:private-cov}
\begin{algorithmic}[1]
\Function{PrivateCov}{$T,\varepsilon,\delta$} \Comment{returns $\Sigma_1^p,\cdots,\Sigma_{T-1}^p$}
	\State $\delta'\gets \frac{\delta}{2\lceil\log_2 T\rceil}$, $\varepsilon' \gets \frac{\varepsilon}{2\lceil\log_2 T\rceil\ln(1/\delta')}$, 
	$\sigma_{\varepsilon',\delta'}^2 = \frac{2\ln(1.25/\delta')}{(\varepsilon')^2}$, $m=\lceil\log_2 T\rceil$;
	\State Initialize $\Sigma(\ell)=\hat\Sigma(\ell)=0$ for all $\ell=0,\cdots,m-1$;
	\For{$n=1,2,\cdots,T-1$}
		\State Express $n$ in its binary form: $n=\sum_{\ell=0}^{m-1}b_n(\ell)2^{\ell}$, $b_n(\ell)\in\{0,1\}$;
		\State Let $\ell_n\gets \min\{\ell: b_n(\ell)=1\}$ be the least significant bit of $n$;
		\State Update $\Sigma(\ell_n) \gets\phi_n\phi_n^\top + \sum_{\ell<\ell_n}\Sigma(\ell)$ and $\Sigma(\ell)\gets\hat\Sigma(\ell)\gets 0$ for all $\ell<\ell_n$;
		\State Calibrate noise: $\hat\Sigma(\ell_n)\gets \Sigma(\ell_n) + W^n$ where $W^n_{ij}=W^n_{ji}\overset{i.i.d.}{\sim}\mathcal N(0, \sigma_{\varepsilon',\delta'}^2)$;
		\State Release $\Sigma_n^p = \sum_{\ell=0}^{m-1}b_n(\ell)\hat\Sigma(\ell)$;
	\EndFor
\EndFunction
\end{algorithmic}
\end{algorithm}

\begin{figure}[t]
	\centering
	\includegraphics[width=0.6\textwidth]{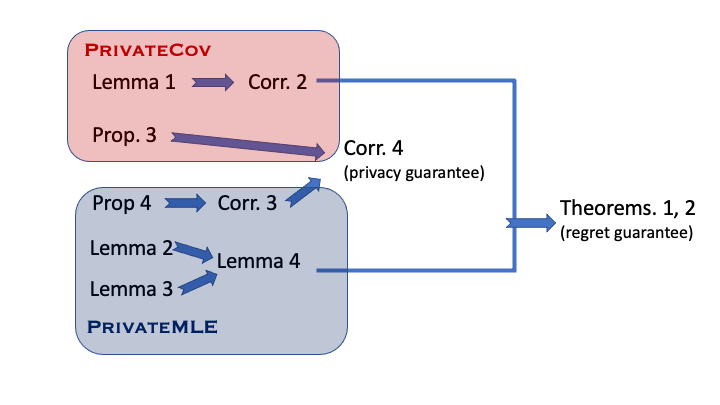}
	\caption{Flow of our proof framework.}
	\label{fig:proof-flow}
\end{figure}\textbf{}

{  Algorithm \ref{alg:private-cov} is based on the AnalyzeGauss framework in \citep{dwork2014analyze} coupled with the \emph{tree-based aggregation} technique for releasing continual observations \citep{dwork2010differential,chan2011private}. The AnalyzeGauss  by \cite{dwork2014analyze} develops a Gaussian mechanism on releasing a single covariance matrix privately from the data.  On the other hand, tree-based aggregation provides a general protocol on how to continually release sequentially updated statistics (e.g., partial sums of sample covariance matrices) under privacy constraints. For our \textsc{PrivateCov}, by calibrating symmetric random Gaussian matrices $\{W^n\}$ into the sample covariances under the tree-based aggregation,  one achieves differential privacy.}
The following proposition claims that the outputs $(\Sigma_1^p,\cdots,\Sigma_{T-1}^p)$ of Algorithm \ref{alg:private-cov} satisfy $(\varepsilon,\delta)$-differential privacy.
\begin{proposition}
The outputs of Algorithm \ref{alg:private-cov}, $(\Sigma_1^p, \ldots, \Sigma_{T-1}^p)$ satisfy $(\varepsilon,\delta)$-differential privacy.
\label{prop:privacy-cov}
\end{proposition}

The following lemma further gives high probability bounds on the deviation from $\Sigma_n^p$ to the actual sample covariance $\Sigma_n=\sum_{t=1}^n\phi_t\phi_t^\top$.
This utility guarantee is useful later in the regret analysis to justify the $\det(\Lambda_n^p)>2\det(\Lambda^p)$ condition in Algorithm \ref{alg:framework}.
\begin{lemma}
With probability $1-O(T^{-1})$, it holds for all $n\in\{1,2,\cdots,T-1\}$ that
$$
\|\Sigma_n^p-\Sigma_n\|_{\mathrm{op}} \leq O(\varepsilon^{-1}\sqrt{d}\ln^{4.5}(T/\delta)),
$$
where $\Sigma_n = \sum_{t\leq n}\phi_t\phi_t^\top$.
\label{lem:utility-cov}
\end{lemma}

The following corollary is an immediate consequence of Lemma \ref{lem:utility-cov}.

\begin{corollary}
Let $\Lambda_n = \Sigma_n + \rho I_d$ and $\Lambda_n^p=\Sigma_n^p+\rho I_d$ for some $\rho\geq \varepsilon^{-1}d\sqrt{d}\ln^5 (T/\delta)$.
Then there exists a universal constant $C_T<\infty$ such that, for any $T\geq C_T$, with probability $1-O(T^{-1})$ for all $n\in\{1,2,\cdots,T-1\}$, 
it holds that $0.9\det(\Lambda_n)\leq \det(\Lambda_n^p)\leq 1.11\det(\Lambda_n)$.
\label{cor:cov-multiplicative-approx}
\end{corollary}

{
Corollary \ref{cor:cov-multiplicative-approx} shows that when the $\textsc{PrivateMLE}$ is invoked in Algorithm \ref{alg:framework}, the determinant of the real sample covariance  matrix roughly doubles. This is important to our later regret analysis because if \textsc{PrivateMLE} is invoked too frequently, the algorithm pays the price of composition in privacy. 
 While on the other hand, if \textsc{PrivateMLE} is invoked too rarely, the old (and inaccurate) parameters will be used for a long time, which incurs a larger regret. Therefore, our analysis shows that the right frequency should be invoking \textsc{PrivateMLE}  once the determinant of the privacy-preserving covariance roughly doubles.
}

\subsection{The \textsc{PrivateMLE} sub-routine}\label{subsec:private-mle}

Algorithm \ref{alg:private-mle} gives a pseudo-code description of the \textsc{PrivateMLE} sub-routine.
The algorithm is based on the ``objective perturbation'' framework developed in \citep{chaudhuri2011differentially,kifer2012private}.
More specifically, Algorithm \ref{alg:private-mle} calibrates a noisy term ($w^\top\theta$) into the constrained maximum likelihood estimation formulation
in order to achieve differential privacy of the output optimal solutions $\hat\theta_n^p$.

The following proposition establishes the claim that Algorithm \ref{alg:private-mle} is $(\varepsilon,\delta)$-differentially private.
\begin{proposition}
The output of Algorithm \ref{alg:private-mle}, $\hat\theta_n^p$, satisfies $(\varepsilon,\delta)$-differential privacy.
\label{prop:privacy-mle}
\end{proposition}

{ The next corollary, which establishes the privacy guarantee of the \textsc{PrivateMLE},} immediately follows Proposition \ref{prop:privacy-mle} and Corollary \ref{cor:advanced-composition}.
It shows how to set the algorithmic parameters in Algorithm \ref{alg:private-mle} to ensure that the resulting price decisions
are differentially private at the designated levels $\varepsilon$ and $\delta$.
\begin{corollary}
Suppose \textsc{PrivateMLE} is invoked for at most $D_\infty$ times in Algorithm \ref{alg:framework}.
Then the composite sequence of $D_{\infty}$ outputs of \textsc{PrivateMLE} satisfies $(\varepsilon,\delta)$-differential privacy
if each call of \textsc{PrivateMLE} is supplied with privacy parameters $\delta'=\frac{\delta}{2D_\infty}$ and $\varepsilon'=\frac{\varepsilon}{\sqrt{2D_\infty\ln(1/\delta')}}$.
\label{cor:privacy-mle}
\end{corollary}

Now, we are ready to provide the privacy guarantee of the entire policy in Algorithm \ref{alg:framework}. 
\begin{corollary}\label{cor:privacy-algo}
The  price decisions $\{p_1, \ldots, p_T\}$ of Algorithm \ref{alg:framework} satisfy $(\varepsilon_1+\varepsilon_2, \delta_1+\delta_2)$-differential privacy.	
\end{corollary}

Corollary \ref{cor:privacy-algo} immediately follows Proposition \ref{prop:privacy-cov}, Corollary \ref{cor:privacy-mle}, Proposition  \ref{prop:na-privacy}, and Fact \ref{fact:composition}. More specifically, in Algorithm \ref{alg:framework},  \textsc{PrivateCov} is invoked with parameters $(\varepsilon_1,\delta_1)$, which is $(\varepsilon_1,\delta_1)$-differential privacy. 
Moreover, \textsc{PrivateMLE} is invoked with
parameters $(\varepsilon_2',\delta_2')$ for at most $D_\infty$ times, whose outputs are $(\varepsilon_2,\delta_2)$-differential privacy. Therefore, the entire policy satisfies $(\varepsilon_1+\varepsilon_2,\delta_1+\delta_2)$-differential privacy,
thanks to Proposition \ref{prop:na-privacy} and the basic composition rule in Fact \ref{fact:composition}.

\begin{algorithm}[!t]
	\caption{The \textsc{PrivateMLE} sub-routine}
	\label{alg:private-mle}
	\begin{algorithmic}[1]
		\Function{PrivateMLE}{$n,\rho,\varepsilon,\delta$} \Comment{returns $\hat\theta_n^p$}
		\State $B_1 \gets (B_Y+1)G$, $B_2\gets KG$, $\rho\gets\max\{\rho,2B_2/\varepsilon\}$, $\nu_{\varepsilon,\delta}^2 \gets B_1^2(8\ln(2/\delta)+4\varepsilon)/\varepsilon^2$;
		\State Sample $w\sim \mathcal N(0, \nu_{\varepsilon,\delta}^2 I_d)$;
		\State Return $\hat\theta_n^p = \arg\min_{\|\theta\|_2\leq 2} \{(\sum_{t<n}-\ln p(y_t|\phi_t,\theta)) + \frac{\rho}{2}\|\theta\|_2^2 + w^\top\theta\}$;
		\EndFunction
	\end{algorithmic}
\end{algorithm}

{ In the rest of this section we establish the prediction error guarantee (a.k.a., the utility guarantee in differential privacy literature) of the estimator $\hat\theta_n^p$ from the \textsc{PrivateMLE} sub-routine.  More precisely, Lemma \ref{lem:mle-utility-3} below upper bounds the   prediction errors of the sequence of obtained
model estimate $\hat\theta_n^p$ with the presence of artificially calibrated noises in the \textsc{PrivateMLE} sub-routine. With smaller values of $\varepsilon,\delta$ indicating stronger privacy protection, the parameter $\nu_{\varepsilon,\delta}$ becomes larger, which leads to a larger variance of the Gaussian noise $w$.  Thus, the \textsc{PrivateMLE} sub-routine needs to calibrate higher magnitudes of noise into the objective function, leading to either larger prediction errors (see \eqref{eq:MLE_upper_1}) or lengthened forced exploration phase (see the condition of  \eqref{eq:MLE_upper_2}). Note that the lower bound on $\lambda_{\min}(\Sigma_n)$ for \eqref{eq:MLE_upper_2} is achieved by the exploration phase in Algorithm \ref{alg:framework}. 
This key result quantifies the tradeoff between the strength of the privacy protection and the prediction errors of the model parameter estimates.  Another practical guidance from Lemma \ref{lem:mle-utility-3} is that the regularization amount $\rho$ also needs to grow as  $\varepsilon,\delta$  becomes smaller.
}

We emphasize that this key utility guarantee in Lemma \ref{lem:mle-utility-3} is  \emph{not} directly covered by existing utility analysis
in \cite{chaudhuri2011differentially,kifer2012private} for two reasons.
First, in \cite{chaudhuri2011differentially,kifer2012private}, the utility is measured in terms of
the difference between \emph{objective values} before and after objective perturbation,
which is \emph{not} sufficient for the purpose of analyzing contextual bandit algorithms that require first-order KKT conditions.
Additionally, in both \cite{chaudhuri2011differentially,kifer2012private}, the data $(\phi_t,y_t)$ are assumed to be sampled
\emph{independently and identically} from an underlying distribution, while in our problem the data clearly are neither independent 
nor identically distributed.

We also remark that our utility analysis of the (differentially private) constrained maximum likelihood estimation (see the proof  of Lemma \ref{lem:mle-utility-3}) differs
significantly from existing analysis of generalized linear contextual bandit problems as well \citep{filippi2010parametric,li2017provably,wang2019optimism}.
In \cite{li2017provably}, it is assumed that $\phi_t$ are \emph{i.i.d.}~and their distributions satisfy a certain non-degenerate assumption,
which we do not necessarily impose in this paper.
In both \cite{filippi2010parametric} and \cite{wang2019optimism}, the formulations of the optimization problems are non-convex in $\theta$, {which facilitates the analysis of the properties of the optimal solution.}   However, the non-convex formulation poses significant challenges for privacy-aware algorithms since differentially private methods for non-convex optimization
are scarce.
It is therefore a highly non-trivial task to analyze a fully convex optimization formulation without stochasticity assumptions on $\phi_t$.

\begin{lemma}
Fix $n\in\{1,2,\cdots,T-1\}$ and let $\Lambda_n = \Sigma_n + \rho I =\sum_{t<n}\phi_t\phi_t^\top + \rho I$.
Suppose $\rho\geq \max\{5\nu_{\varepsilon,\delta}\sqrt{5d\ln T},2 + 48s^2G^2Kd\ln T\}$.
Then with probability $1-O(T^{-2})$ the following hold: $\|\hat\theta_n^p\|_2<2$, and
\begin{equation}\label{eq:MLE_upper_1}
(\hat\theta_n^p-\theta^*)^T \Lambda_n(\hat\theta_n^p-\theta^*) \leq \big(sK\sqrt{3d\ln T} + (2G+3)\sqrt{\rho} + G\nu_{\varepsilon,\delta}\sqrt{5d\ln T}\big)^2.
\end{equation}
Furthermore, if $\lambda_{\min}(\Sigma_n)\geq\lambda_0 = [\frac{(2G+3)\rho}{\sqrt{5d\ln T}} + \nu_{\varepsilon,\delta}G]^2$, then the above inequality can be strengthened to
\begin{equation}\label{eq:MLE_upper_2}
(\hat\theta_n^p-\theta^*)^T \Lambda_n(\hat\theta_n^p-\theta^*) \leq [4sK\sqrt{d\ln T}]^2.
\end{equation}
\label{lem:mle-utility-3}
\end{lemma}
Lemma \ref{lem:mle-utility-3} is proved by analyzing the first-order KKT condition at $\hat\theta_n^p$, and is deferred to the supplementary material.
Lemma \ref{lem:mle-utility-3} upper bounds the transformed estimation error of the differentially private MLE $\hat\theta_n^p$ in two upper bounds.
The first upper bound in \eqref{eq:MLE_upper_1} applies to the general setting and has a $G\nu_{\varepsilon,\delta}\sqrt{5d\ln T}$ additive term involving the differential privacy parameters
$\varepsilon$, $\delta$. in the upper bound.
The second upper bound in \eqref{eq:MLE_upper_2}, on the other hand, shows that if the sample covariance matrix $\Sigma_n$ is spectrally lower bounded, then the upper bound on $\|\hat\theta_n^p-\theta^*\|_{\Lambda_n}^2$ can be much improved with only the standard $O(\sqrt{d\ln T})$ term.

\section{Regret analysis}\label{sec:regret}

Section \ref{sec:privacy} has established the privacy guarantees of our dynamic personalized pricing policy (see Corollary \ref{cor:privacy-algo}). In this section, we will further analyze the \emph{performance/utility} of our proposed policy by proving upper bounds on its expected cumulative \emph{regret}.

Recall that in the dynamic personalized pricing problem, there are $t$ time periods and at each time period a customer arrives with personal information $x_t$.
When offered price $p_t$, the expected demand is modeled by the generalized linear model $p(y_t|p_t,x_t,\theta^*) = \exp\{\zeta(y\phi(p_t,x_t)^\top\theta^* - m(\phi(p_t,x_t)^\top\theta^*)+h(y_t,\zeta)\}$ with expectation $\mathbb E[y_t|p_t,x_t,\theta^*]=f(\phi(p_t,x_t)^\top\theta^*)$.
With $\theta^*$ known in hindsight, the optimal price $p_t^*$ at time $t$ is the one maximizing the retailer's expected revenue, or more specifically
$$
p_t^* := \arg\max_{p\in[0,1]} pf(\phi(p,x_t)^\top\theta^*).
$$
The \emph{regret} of a dynamic pricing policy $\pi$ is then defined as the cumulative difference between the expected revenue of the policy's offered prices
and that of a clairvoyant, or more specifically
$$
\regret(\pi;T) := \sum_{t=1}^T p_t^*f(\phi(p_t^*,x_t)^\top\theta^*) - p_tf(\phi(p_t,x_t)^\top\theta^*).
$$
Clearly, by definition, the regret of any admissible policy is always non-negative since no $p_t$ has a higher expected revenue compared to $p_t^*$.
The smaller the regret, the better the policy's performance.
We are also primarily focused on the \emph{asymptotic} growth of the regret as a function of the time horizon $T$, as well as several other important parameters, such as the feature dimension $d$ and the privacy parameters $\varepsilon_0:=\varepsilon_1+\varepsilon_2$, $\delta_0:=\delta_1+\delta_2$.

\subsection{The general case}
\label{sec:general_case}
We first analyze the regret of Algorithm \ref{alg:framework} in the most general case, in which customers' personal information $\{x_t\}$ is obliviously {(i.e., pre-fixed)} but  can be adversarially chosen without pre-assumed patterns.
Our next theorem upper bounds the regret of Algorithm \ref{alg:framework} with proper choices of the values of algorithmic parameters. {Recall that $\varepsilon_0:=\varepsilon_1+\varepsilon_2$, $\delta_0:=\delta_1+\delta_2$. We also note that for the general case, the random exploration phase (Step 4 in Algorithm \ref{alg:framework}) will be unnecessary and thus we could set $T_0=0$.}
\begin{theorem}
Suppose Algorithm \ref{alg:framework} is run with parameters $\varepsilon_1,\varepsilon_2\geq 0.1\varepsilon_0$, $\delta_1,\delta_2\geq 0.1\delta_0$, 
$T_0=0$, $D_{\infty} = \lceil d\log_{1.5} T\rceil$, $\rho = \max\{\varepsilon_1^{-1}d^{1.5}\ln^5 T, 5\nu_{\varepsilon_2',\delta_2'}\sqrt{5d\ln T}, 2+48s^2G^2Kd\ln T\}$, 
$\gamma = K[(\sqrt{3}sK+\sqrt{5}G\nu_{\varepsilon_2',\delta_2'})\sqrt{d\ln T}+(2G+3)\sqrt{\rho}]$, 
where $\varepsilon_2',\delta_2'$ are defined in Step 3 of Algorithm \ref{alg:framework}
and $\nu_{\varepsilon_2',\delta_2'}$ is defined in Algorithm \ref{alg:private-mle}.
Then it holds that
$$
\regret(\pi;T) \leq 2\gamma\sqrt{4.6dT\ln T} \leq \widetilde O\left(\varepsilon_0^{-1}{\sqrt{d^3 T\ln^5(1/\delta_0)}}\right),
$$
where in the $\widetilde O(\cdot)$ notation we omit logarithmic terms in $T$ and polynomial dependency on other model parameters $s,K,G$ and $B_Y$.
\label{thm:main-regret-adversarial-context}
\end{theorem}

Theorem \ref{thm:main-regret-adversarial-context} is proved in the supplementary material.
We note that when $T$ is large, our regret bound matches the classical optimal regret bound of $O(\sqrt{T})$. The dependency on the dimensionality of personal information $d$ (i.e., $\sqrt{d^3}$) can be further improved by assuming a stronger assumption on the stochasticity of personal information ${x_t}$ (see Section \ref{subsec:improved}). Stochastic personal information or demand covariate  has been a common assumption in the pricing literature (see e.g., \citealt{qiang2016dynamic,ban2017personalized,javanmard2019dynamic,chen2020statistical}). 



\subsection{Improved regret with stochastic contexts}
\label{subsec:improved}
In this section, we show that for a large class of problems in which the customers' personal information is \emph{stochastically distributed},
the regret upper bound in Theorem \ref{thm:main-regret-adversarial-context} could be significantly sharpened.

The following assumption mathematically characterizes the stochasticity condition of customers' personal information used in this section:
\begin{assumption}
Let $U[0,1]$ be the uniform distribution on $[0,1]$.
There exists an underlying distribution $\mu_x$ and a constant $\kappa_x>0$ such that, $x_1,\cdots,x_T\overset{i.i.d.}{\sim}\mu_x$,
and furthermore
$$
\|\phi(x,p)\|_2\leq 1\;\;\;\;a.s.\;\;\sim\mu_x\times U[0,1]; \;\;\;\;\;\; \mathbb E_{(x,p)\sim\mu_x\times U[0,1]}\big[\phi(p,x)\phi(p,x)^\top\big]\succeq \kappa_x I_d.
$$
\label{asmp:stochastic-context}
\end{assumption}

{
Assumption \ref{asmp:stochastic-context} assumes that consumers' personal feature vectors are relatively widely spread,
so that they are not concentrated in a narrow region or direction.
Such an assumption helps improve the regret analysis because the algorithm can expect to see feature vectors along with any directions
with reasonable chances, and therefore the overall estimates of the unknown regression model can be more accurate.
}

With Assumption \ref{asmp:stochastic-context}, the following theorem shows that when algorithmic parameters are properly chosen
in Algorithm \ref{alg:framework}, the regret upper bound can be improved compared to Theorem \ref{thm:main-regret-adversarial-context} for the general setting.
\begin{theorem}
Under Assumption \ref{asmp:stochastic-context}, suppose Algorithm \ref{alg:framework} is run with parameters $\varepsilon_1,\varepsilon_2\geq 0.1\varepsilon_0$, $\delta_1,\delta_2\geq 0.1\delta_0$, $D_{\infty} = \lceil d\log_{1.5} T\rceil$, 
$\rho = \max\{\varepsilon_1^{-1}d^{1.5}\ln^5 T, 5\nu_{\varepsilon_2',\delta_2'}\sqrt{5d\ln T}, 2+48s^2GKd\ln T\}$, 
$T_0 = 32 [\frac{(2G+3)\rho}{\sqrt{5d\ln T}} + \nu_{\varepsilon,\delta}G]^2\ln^2(dT)$, 
$\gamma = 4sK^2\sqrt{d\ln T}$, 
where $\varepsilon_2',\delta_2'$ are defined in Step 3 of Algorithm \ref{alg:framework}
and $\nu_{\varepsilon_2',\delta_2'}$ is defined in Algorithm \ref{alg:private-mle}.
Then it holds for sufficiently large $T\geq e^{\kappa_x^{-2}}$ that
$$
\regret(\pi,T) \leq T_0 +  2\gamma\sqrt{4.6dT\ln T} \leq \widetilde O\left(d \sqrt{T} + \varepsilon_0^{-2}d^2\ln^{10}(1/\delta_0)\right),
$$
where in the $\widetilde O(\cdot)$ notation we omit logarithmic terms in $T$ and polynomial dependency on other model parameters $s,K,G$ and $B_Y$.
\label{thm:main-regret-stochastic-context}
\end{theorem}

The proof of Theorem \ref{thm:main-regret-stochastic-context} is largely the same as the proof of Theorem \ref{thm:main-regret-adversarial-context},
except for the application of the second upper bound in Lemma \ref{lem:mle-utility-3}.
We relegate the complete proof of Theorem \ref{thm:main-regret-stochastic-context} to the supplementary material.
Comparing Theorem \ref{thm:main-regret-stochastic-context} with Theorem \ref{thm:main-regret-adversarial-context},
we note that the significant improvement lies in the additive nature between $\varepsilon_0, \delta_0$ and $d, T$ terms in Theorem \ref{thm:main-regret-stochastic-context}.
More specifically, because the privacy-incurred terms are now additive and do not scale polynomially with $T$,
in most practical scenarios when the time horizon $T$ is very large, the dominating term of Theorem \ref{thm:main-regret-stochastic-context}
becomes only $\widetilde O(d\sqrt{T})$, which is optimal (up to logarithmic factors) in \emph{both} the time horizon $T$
and the feature dimension $d$ (see, for example, the $\Omega(d\sqrt{T})$ lower bound in \cite{dani2008stochastic}). 

\subsection{Impact of privacy constraints on seller surplus}\label{subsec:discussion-seller-surplus}

In our theoretical framework, the seller surplus is measured and reflected by the notion of \emph{regret},
which measures how much revenue/profits are lost by the seller's pricing decisions compared to the optimal personalized prices in hindsight.
The smaller the regret, the larger the seller surplus.

Our main results in Theorems \ref{thm:main-regret-adversarial-context} and \ref{thm:main-regret-stochastic-context}
give quantitative upper bounds on the regret of our proposed algorithm.
More specifically, the regret of our algorithm is (omitting logarithmic factors and secondary model parameters)
$\widetilde O(\varepsilon^{-1}\sqrt{d^3 T})$ in the general setting,
and $\widetilde O(\sqrt{d^2 T}+\varepsilon^{-2}d^2)$ with additional assumptions on the distribution of consumers' context vectors.
Here $T$ is the time horizon (i.e., the number of customers handled), $d$ is the number of covariates in consumers' personal data, and $\varepsilon>0$ dictates the level of privacy leakage, with smaller $\varepsilon$ indicating stronger/stricter protection of users' privacy.
Based on these results, we make the following observations:

\paragraph{Tradeoffs between seller profits and privacy protection.} With stronger privacy protection (i.e., $\varepsilon\to 0^+$), 
it is clear that the regret of our proposed algorithm increases, indicating that the seller profits are going to suffer
with additional privacy constraints.
The decrease of seller surplus is, however, alleviated when the consumers' context vectors are relatively well distributed, 
as the $\varepsilon^{-2}d^2$ term is \emph{not} the dominating term in the regret bound when there are sufficient number of customers/users.
Such decrease of seller profits is intuitive and expected, because additional privacy constraints limit sellers' ability to offer very personally 
tailored prices to boost their revenues.

\paragraph{Value and privacy costs of information.} 
The $d$ parameter in the regret bound characterizes how many covariates or factors the pricing algorithm exploits in customers' personalized data and shows the value and privacy costs of information: with more factors/covariates (i.e., larger values of $d$),
the retailer is able to consider more refined details and information of each incoming customer
but such information adds to the burden of privacy protection, leading to increased regret. {To see this more clearly, with some stochasticity assumption of covariates, the regret bound $\widetilde O\left(d \sqrt{T} + \varepsilon^{-2}d^2 \right)$ in Theorem \ref{thm:main-regret-stochastic-context} shows the following fact. For the regret term $\varepsilon^{-2}d^2 $ related to the privacy to be a constant,  a larger dimension $d$ (i.e., more customer information) implies that $\varepsilon=C_0 d$ also grows proportionally, which leads to a weaker privacy protection. 
{ Additionally, for the first term $\widetilde O(d\sqrt{T})$, there is also a known lower bound showing that any policy must suffer a regret of $\Omega(d\sqrt{T})$
in the worst case \citep{dani2008stochastic}.
Therefore, there is indeed a cost of information for the purpose of privacy protection.}
Our regret upper bounds therefore provide in principle a bottom line for practitioners to gauge the costs of incorporating more factors
of user information into dynamic personalized price decisions.

\section{Numerical results}\label{sec:numerical}

In this section we corroborate the theoretical guarantees established in this paper for our proposed differentially private
personalized pricing method with simulation results on a synthetic dataset.
We adopt the logistic regression model $\Pr[y_t=1|\phi_t,\theta^*] = \frac{e^{\zeta\phi_t^\top\theta^*}}{1+e^{\zeta\phi_t^\top\theta^*}}$,
with $\zeta=4$, $\phi_t(x_t,p_t) = \frac{1}{\sqrt{d}}[x_t; -p_t]\in\mathbb R^d$
and $\theta^*=[-\sqrt{0.1}; -\sqrt{0.1}; \cdots; -\sqrt{0.1}; \sqrt{1-0.1(d-1)}]\in\mathbb R^d$.
The personal feature vectors $\{x_t\}$ are synthesized uniformly at random from the unit cube $[-1,1]^{d-1}$.
It is easy to verify that $\|\phi_t\|_2\leq 1$ and $\|\theta^*\|_2\leq 1$ always hold for all $d$.
Algorithm parameters (as inputs in Algorithm \ref{alg:framework}) are chosen as $T_0=10$, $\rho=10$, $D_{\infty}=\lceil d\log_2 T\rceil$, and $\gamma=1$.
Other privacy-related parameters will be varied to demonstrate a  spectrum of our proposed algorithm on a continuous landscape of differential privacy guarantees. {Note that this experiment's main purpose is to investigate the impact of privacy-related parameters (i.e., $\varepsilon$ and $\delta$) rather than compete with state-of-the-art non-private pricing algorithms.} 

\begin{figure}[t]
\centering
\includegraphics[width=0.48\textwidth]{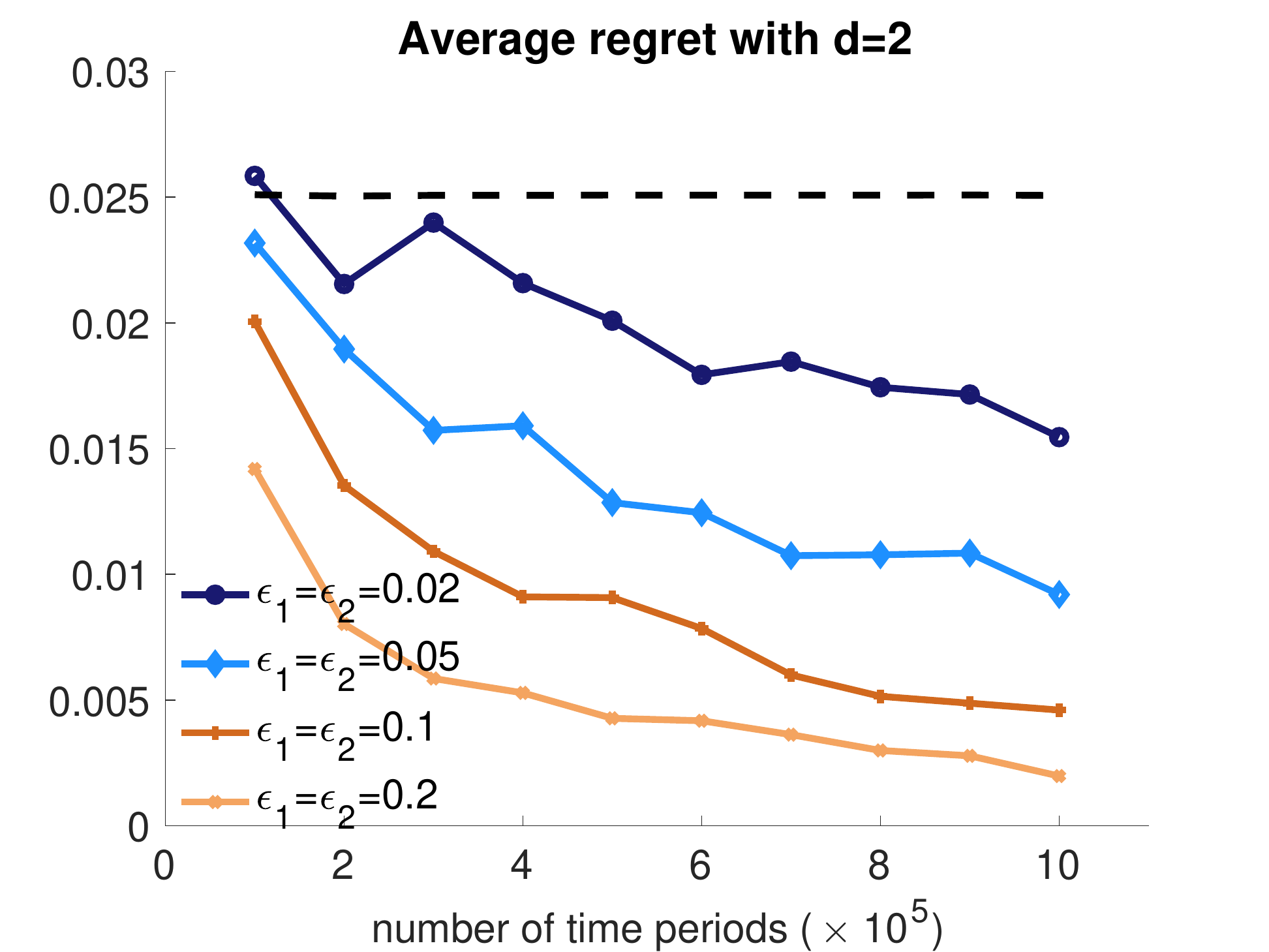}
\includegraphics[width=0.48\textwidth]{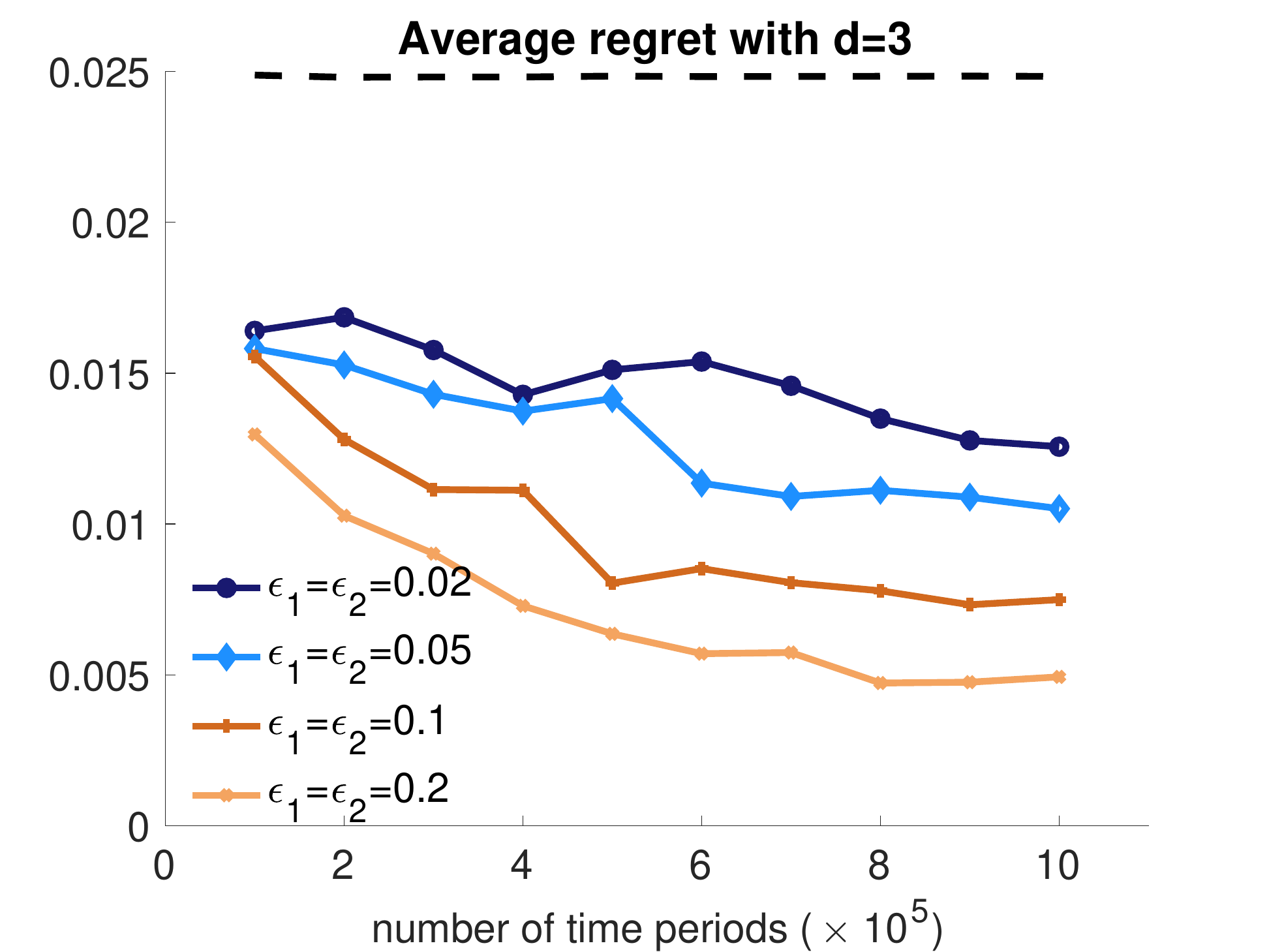}
\caption{Average regret of our proposed algorithm under different time horizons $T$.
The black dashed line indicates the average regret of a policy offering completely at random prices.
Both $\delta_1,\delta_2$ parameters are set at $\delta_1=\delta_2=1/T^2$.}
\label{fig:regret-main}
\end{figure}

In Figure \ref{fig:regret-main} we plot the average regret of our proposed algorithm under various $\varepsilon_1,\varepsilon_2$ privacy settings
and time horizons $T$ ranging from $10^5$ to $10^6$.
All settings are run for 20 independent trials and the average regret is reported.
For reference purposes, we also indicate in both plots of Figure \ref{fig:regret-main} (see the flat dashed line) the average regret of a policy that simply produces
uniformly at random prices $p_t$ at each $t$, completely ignoring the personalized features/factors of each incoming customer.
As we can see, under most privacy settings including  highly secured settings with small $\varepsilon$ (e.g., $\varepsilon_1=\varepsilon_2=0.02$),
the average regret of our proposed algorithm is much smaller compared to completely random prices, demonstrating its utility
under privacy constraints.
Furthermore, with relaxed privacy requirements (i.e., larger values of $\varepsilon_1,\varepsilon_2$)
and/or longer pricing horizons $T$, the average regret of our algorithm significantly decreases,
which verifies the theoretical regret upper bounds we established in Theorems \ref{thm:main-regret-adversarial-context} and \ref{thm:main-regret-stochastic-context}.

\begin{figure}[t]
\centering
\includegraphics[width=0.48\textwidth]{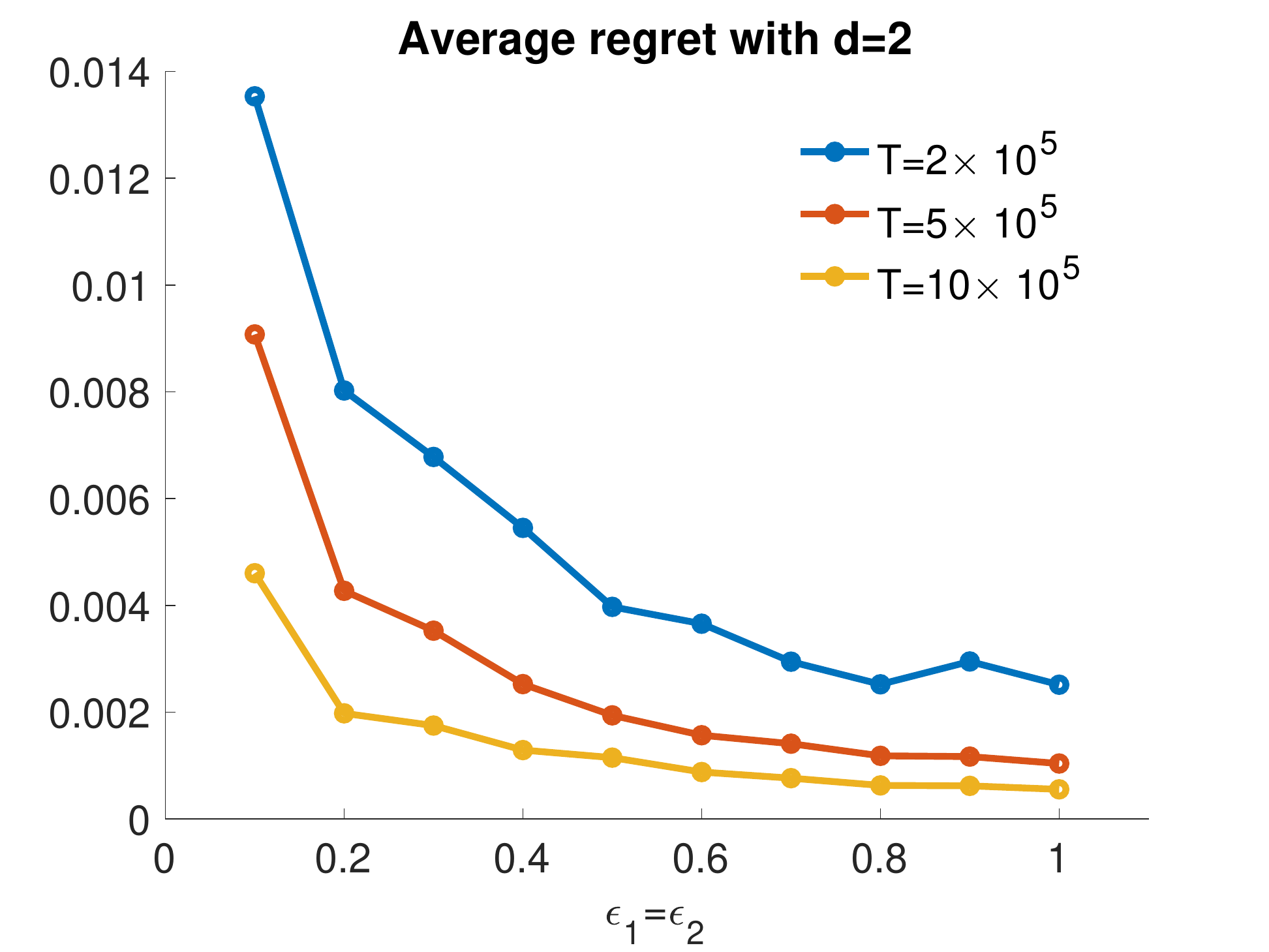}
\includegraphics[width=0.48\textwidth]{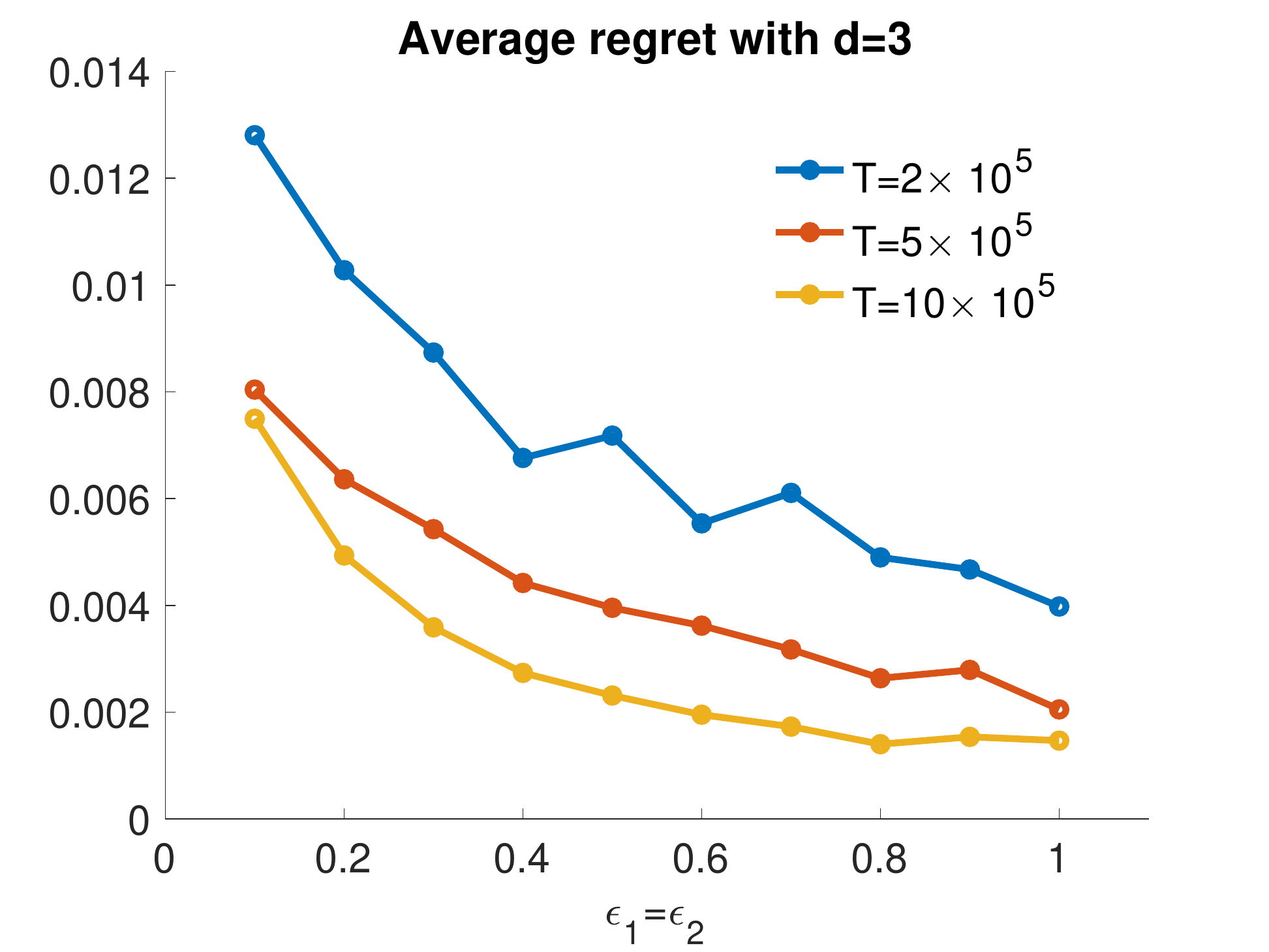}
\caption{Average regret of our proposed algorithm under different privacy parameters $\varepsilon=\varepsilon_1=\varepsilon_2$.
Both $\delta_1,\delta_2$ parameters are set at $\delta_1=\delta_2=1/T^2$.}
\label{fig:regret-veps}
\end{figure}

\begin{figure}[t]
\centering
\includegraphics[width=0.48\textwidth]{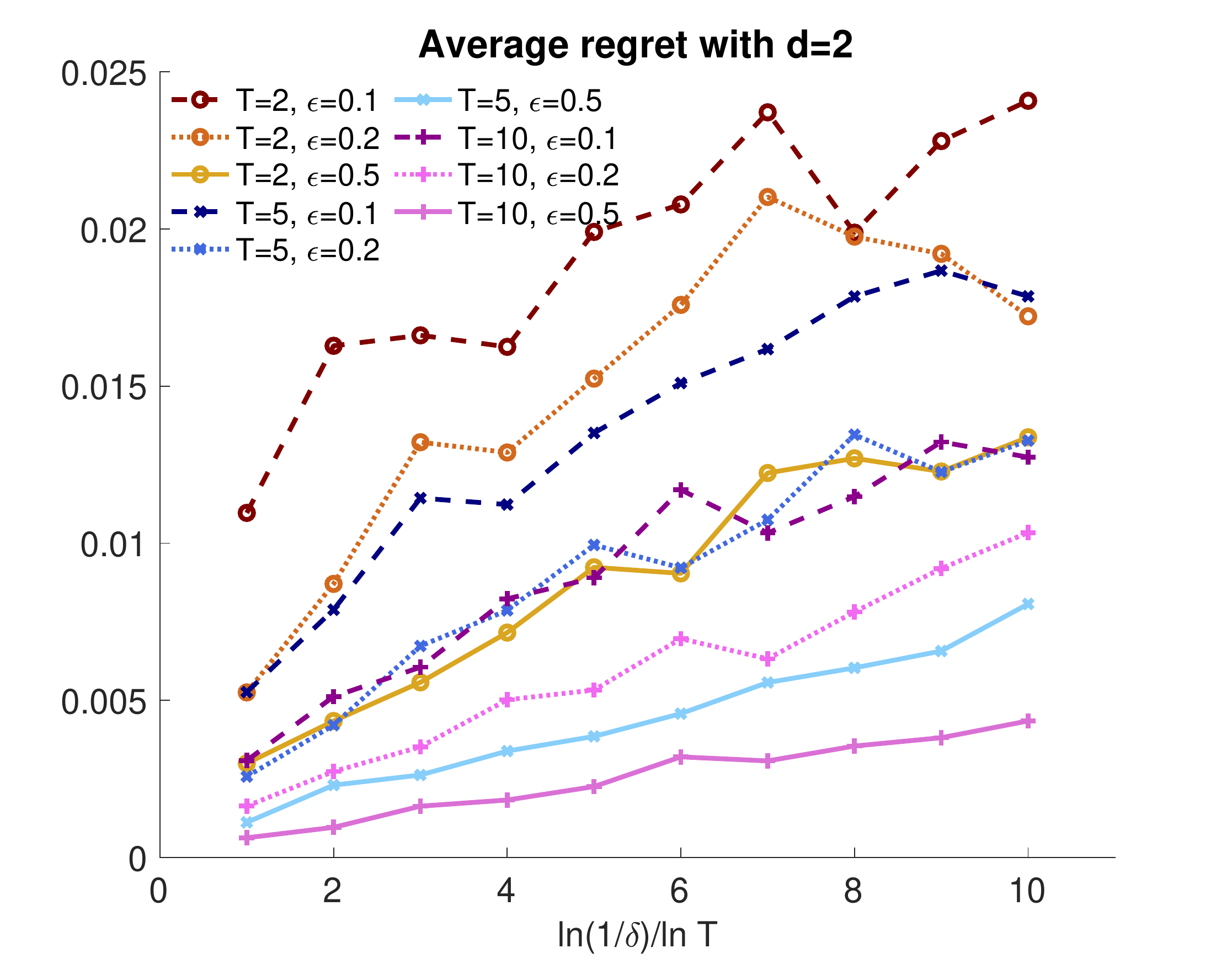}
\includegraphics[width=0.48\textwidth]{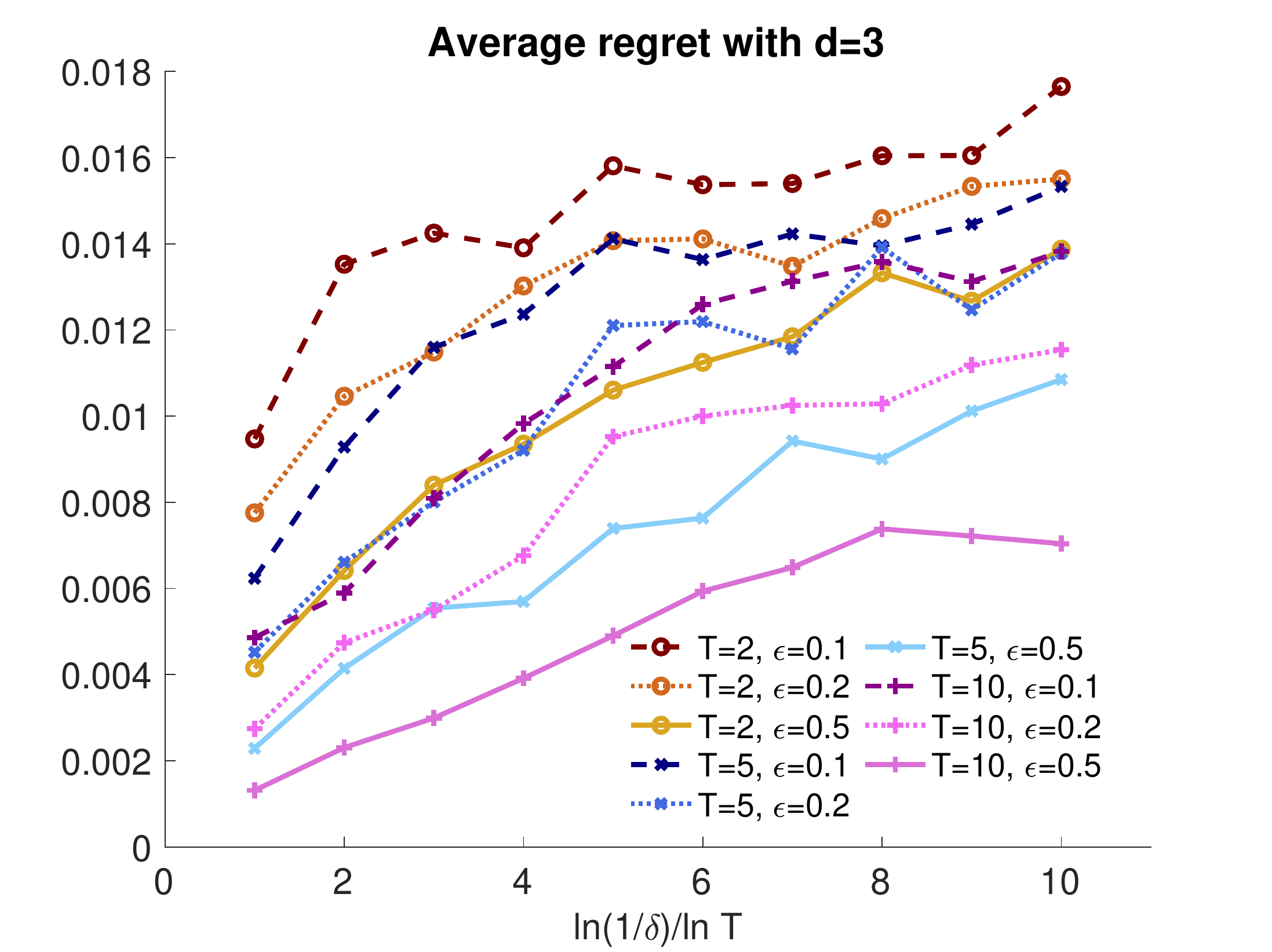}
\caption{Average regret of our proposed algorithm under different $\delta$ parameter values.
From left to right the $\delta$ values are $1/T, 1/T^2,\cdots,1/T^{10}$. The time horizon is measured in terms of $10^5$ periods (i.e., $T=2$ means $2\times 10^5$ total time periods).}
\label{fig:regret-vdelta}
\end{figure}

In Figures \ref{fig:regret-veps} and \ref{fig:regret-vdelta}, we provide some additional auxiliary simulation results.
Figure \ref{fig:regret-veps} gives a direct landscape of the average regret of our algorithm under $\varepsilon$ values
ranging from 0.1 to 1.
Figure \ref{fig:regret-vdelta} further explores the robustness of our algorithm under several very small $\delta$ values
(as small as $\delta=1/T^{10}$).
Note that in Figure \ref{fig:regret-vdelta} there are multiple trend lines corresponding to the performances of the proposed algorithm
under different settings of $T,\varepsilon$ and $\delta$ values.
Apart from the dependency on $\ln(1/\delta)$, Figure \ref{fig:regret-vdelta} also shows that
the average regret of our algorithm decreases with increasing time horizon $T$ and relaxed privacy guarantees (i.e., larger values of $\varepsilon$),
both of which are consistent with the findings in Figures \ref{fig:regret-main} and \ref{fig:regret-veps}.
The results in both figures are as expected (significant decreases in average regret with large $\varepsilon$ values
and moderate increases in average regret with geometrically decreasing $\delta$ values)
from our theoretical results.

\begin{table}[t]
\centering
\caption{Average regret comparison with non-private pricing algorithms.}
\begin{tabular}{l|cccccc}
\hline
& $\varepsilon=0.1$& $\varepsilon=0.2$& $\varepsilon=0.5$& $\varepsilon=1.0$& $\varepsilon=5.0$& non-private\\
\hline
$T=10^5, d=2$& $201\times 10^{-4}$& $142\times 10^{-4}$& $74.6\times 10^{-4}$& $41.9\times 10^{-4}$& $44.7\times 10^{-4}$& $3.1\times 10^{-4}$ \\
$T=10^6, d=2$& $46.0\times 10^{-4}$& $19.8\times 10^{-4}$& $11.5\times 10^{-4}$& $5.6\times 10^{-4}$& $5.5\times 10^{-4}$& $0.6\times 10^{-4}$\\
$T=10^5, d=3$& $156\times 10^{-4}$& $130\times 10^{-4}$& $92.6\times 10^{-4}$& $62.9\times 10^{-4}$& $43.4\times 10^{-4}$& $3.1\times 10^{-4}$\\
$T=10^6, d=3$& $74.9\times 10^{-4}$& $49.4\times 10^{-4}$& $23.1\times 10^{-4}$& $14.7\times 10^{-4}$& $5.7\times 10^{-4}$& $1.6\times 10^{-4}$\\
\hline
\end{tabular}
\label{tab:non-private}
\end{table}

{
To better illustrate our algorithm, we further report two additional sets of simulation results.
In Table \ref{tab:non-private} we report the average regret of our proposed algorithm together with an algorithm that is not subject to any kind 
of privacy constraints,
which is implemented by removing all noise calibration steps in the two private releasers \textsc{PrivateCov} and \textsc{PrivateMLE}.
We remark that even larger $\varepsilon$ values (e.g., $\varepsilon=1.0$ or $\varepsilon=5.0$) indicate
quite non-trivial universal privacy protection of consumers' sensitive data,
which explains the relatively larger regret incurred by differentially private pricing algorithms compared with their non-private counterparts.
In Table \ref{tab:diffeps} we report the average regret of our proposed algorithm when the values of $\varepsilon_1$ and $\varepsilon_2$
are very different to see which privacy parameter has a bigger impact on the performance of the designed algorithm.
Table \ref{tab:diffeps} shows that $\varepsilon_2$ clearly has a much larger effect on the regret performance of our algorithm,
with the average regret significantly decreasing with larger $\varepsilon_2$ values.
On the other hand, the impact of $\varepsilon_1$ is not significant or clear.
This is expected from the structure of the algorithm, because $\varepsilon_2$ is used in the \textsc{PrivateMLE}
sub-routine, which directly affects the model estimates used in subsequent price offerings.
}

\begin{table}[t]
\centering
\caption{Average regret with $T=10^5$ of our algorithm under different $\varepsilon_1,\varepsilon_2$ settings. When the row indicates ``fix $\varepsilon_1\equiv 0.1$'' (or ``fix $\varepsilon_2\equiv 0.1$''), then the $\varepsilon$ in the column represents the value of $\varepsilon_2$ (or accordingly $\varepsilon_1$).}
\begin{tabular}{l|ccccc}
\hline
& $\varepsilon=0.02$& $\varepsilon=0.05$& $\varepsilon=0.1$& $\varepsilon=0.2$& $\varepsilon=0.5$\\
\hline
$T=10^5, d=2$, fix $\varepsilon_1\equiv 0.1$& 0.0247& 0.0232& 0.0195& 0.0145& 0.0073\\
$T=10^5, d=2$, fix $\varepsilon_2\equiv 0.1$& 0.0192& 0.0178& 0.0196& 0.0192& 0.0191 \\
$T = 10^5, d=3$, fix $\varepsilon_1\equiv 0.1$& 0.0164& 0.0160& 0.0147& 0.0128& 0.0092\\
$T=10^5,d=3$, fix $\varepsilon_2\equiv 0.1$& 0.0145& 0.0149& 0.0144& 0.0149& 0.0154\\
\hline
\end{tabular}
\label{tab:diffeps}
\end{table}

{
\section{Discussion and insights}\label{sec:discussion}

\subsection{Insufficiency of input perturbation}\label{subsec:input-perturbation}

\emph{Input perturbation} is a straightforward method for designing differentially private algorithms
and is actually an effective method in some application scenarios.
The high-level idea of input perturbation is to artificially calibrate noise directly to the \emph{inputs} of the algorithm in order to protect private information.
With noisy inputs, the privacy of the entire algorithm trivially follows from the closeness-to-post-processing property of differential privacy (Fact \ref{fact:post-processing}).

In the context of personalized dynamic pricing, application of the input perturbation method amounts to calibrating noise 
directly to the personal features $x_t$ of each incoming customer: $\tilde x_t=x_t+\omega_t$,
for some centered noise vectors $\{\omega_t\}_{t=1}^T$.
Such an approach, however, is likely to fail because the features of each individual customer are relatively independent from each other.
Therefore, a very large magnitude of noises $\{\omega_t\}$ need to be injected, which  renders the subsequent pricing algorithm impractical.
More detailed discussion follows:
\begin{enumerate}
\item Suppose $\tilde x_t=x_t+w_t$ is the anonymized version of a customer's feature vector $x_t$ at time $t$.
Because $\tilde x_t$ is released and used in the subsequent process of the pricing algorithm, one must ensure that $\tilde x_t$ is differentially private.
This means that the magnitude of $w_t$ must be sufficiently large (on the order of $\Omega(1/\varepsilon)$) to protect the sensitive information of $x_t$.
\item Usually, input perturbation results in a much worse performance of the differentially private algorithms compared to output perturbation.
Consider the very simple example of having sensitive data $x_1,\cdots,x_n$ and one wants to release $\bar x=\frac{1}{n}\sum_{i=1}^n x_i$ with $\varepsilon$-differential privacy.
If we use input perturbation with $\tilde x_i=x_i+w_i$ and the Laplace mechanism, we have $w_i\sim\mathrm{Lap}(0,1/\varepsilon)$ and therefore
$\tilde x^1 := \frac{1}{n}\sum_{i=1}^n \tilde x_i$ satisfies $\mathbb E[|\tilde x^1-\bar x|]\asymp O(1/\varepsilon\sqrt{n})$.
On the other hand, if one uses output perturbation by releasing $\tilde x^2 := \bar x + \frac{1}{n}\mathrm{Lap}(0,1/\varepsilon)$, then one has
$\mathbb E[|\tilde x^2-\bar x|]\asymp O(1/\varepsilon n)$.
It is easy to verify that both $\tilde x^1,\tilde x^2$ are differentially private, but $\tilde x^2$ clearly is much closer to $\bar x$ compared to $\tilde x^1$.
This very simple example shows that, in general, input perturbation (directly adding noises to sensitive data) is usually less efficient and should be avoided if there 
are better approaches.
\item In the particular model studied in this paper, the use of a generalized linear model further complicates the input perturbation-based methods.
For many generalized linear models, such as the logistic regression model, the efficiency of statistical estimates (e.g., the maximum likelihood estimation)
decays \emph{exponentially} fast with respect to the vector norm of the feature vector $x$.
Hence, if we use $\tilde x_t=x_t+w_t$ to replace $x_t$ directly in the logistic regression model, the norm of $\tilde x_t$ is on the order of $\Omega(1/\varepsilon)$
and therefore the resulting method is going to incur an $O(\exp\{1/\varepsilon\})$ term in regret, which makes the regret excessively large.
\end{enumerate}

\begin{table}[t]
\centering
\caption{Average regret of our proposed algorithm and the input perturbation method.}
\begin{tabular}{lcccccccc}
\hline
& &\multicolumn{3}{c}{our algorithm}&& \multicolumn{3}{c}{input perturbation}\\
\cline{3-5} \cline{7-9}
& & $\varepsilon=0.2$& $\varepsilon=0.5$& $\varepsilon=1.0$& & $\varepsilon=0.2$& $\varepsilon=0.5$& $\varepsilon=1.0$ \\
\hline
$T=10^5, d=2$ ($\times 10^{-4}$)& &142&  74.6&  41.9& & 393& 393& 98.2 \\
$T=5\times 10^5, d=2$ ($\times 10^{-4}$)& &42.7& 19.4& 10.4& & 393& 393& 95.8 \\
$T=10^6, d=2$ ($\times 10^{-4}$)& & 19.8& 11.5& 5.6& & 393& 393& 95.3\\
\hline
\end{tabular}
\label{tab:input-perturbation}
\end{table}

In Table \ref{tab:input-perturbation} we compare the average regret of our proposed algorithm with the input perturbation method
using numerical simulations.
Table \ref{tab:input-perturbation} shows that the regret of our designed algorithm is significantly smaller than that of the input perturbation.
Furthermore, the average regret of input perturbation is very large unless the $\varepsilon$ parameter is at least one and does not necessarily decrease with increasing number of time periods $T$.

\subsection{Impact of privacy constraints on consumer surplus}\label{subsec:discussion-consumer-surplus}

\begin{figure}[t]
\centering
\includegraphics[width=0.48\textwidth]{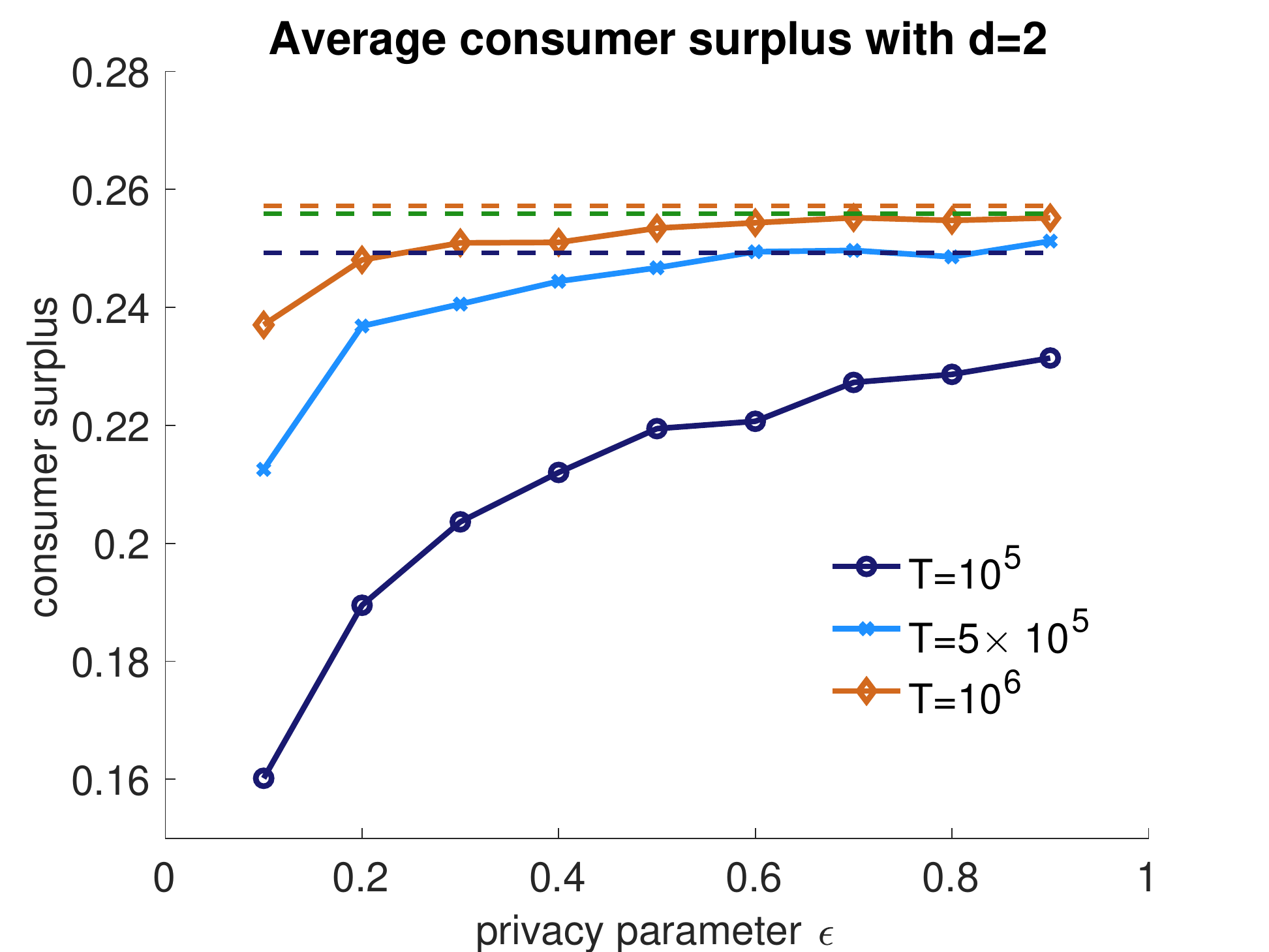}
\includegraphics[width=0.48\textwidth]{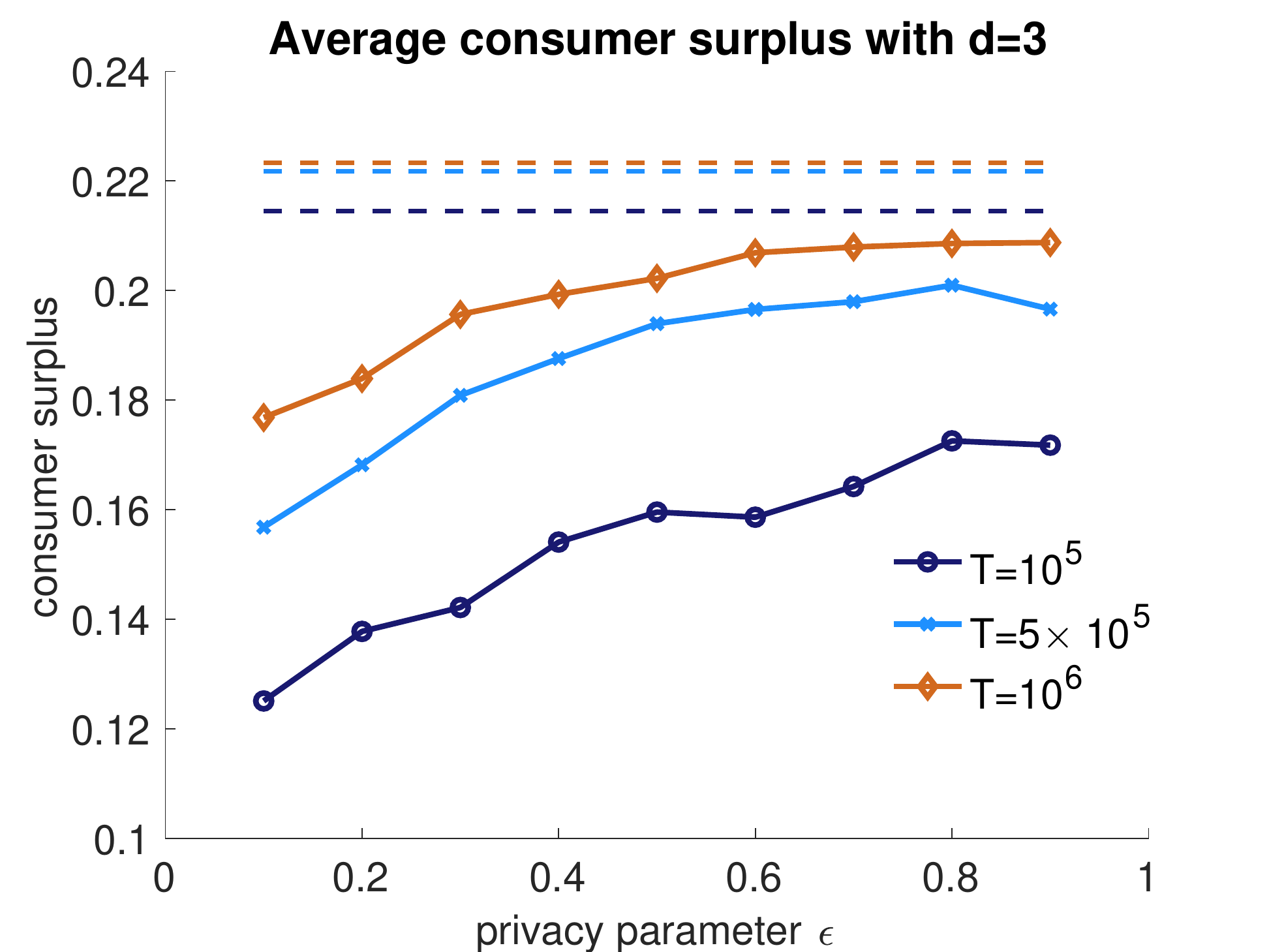}
\caption{Average consumer surplus under different levels of privacy constraints and time horizons.
The dashed lines represent average consumer surplus for a personalized pricing algorithm \emph{not} subject to any data privacy constraints.
Both $\varepsilon_1,\varepsilon_2$ parameters are equal to $\epsilon$ in the figures, and both $\delta_1,\delta_2$ parameters
are set as $1/T^2$, where $T$ is the time horizon.}
\label{fig:surplus}
\end{figure}

In this section we study the impact of privacy constraints of the seller's personalized pricing algorithm
on the average consumer's surplus, under different levels of privacy constraints.
We model the utility $u_t$ for each incoming customer at time $t$ with feature vector $x_t$ and offered price $p_t$ as
$u_t = \zeta \langle\phi(x_t,p_t), \theta^*\rangle + \zeta_t$, where $\zeta=4$, $\phi(x_t,p_t) = \frac{1}{\sqrt{d}}[x_t;p_t]$
and $\zeta_t$ are i.i.d.~random variables following the standard centered Logistic distribution.
The customer will make one unit of purchase if $u_t>0$, resulting in a surplus of $u_t$, and leave without making any purchases if $u_t<0$,
resulting in zero surplus at that period.
It is easy to verify that this utility model leads to the logistic regression model we used in the numerical experiments,
or more specifically, $\Pr[y_t=1|x_t,p_t] = \Pr[u_t>0|x_t,p_t] = \frac{e^{\zeta\phi_t^\top\theta^*}}{1+e^{\zeta\phi_t^\top\theta^*}}$,
where $\phi_t=\phi(x_t,p_t)$.

Figure \ref{fig:surplus} reports the average consumer surplus under our proposed privacy-aware personalized pricing algorithm,
for both $d=2$ and $3$ with consumers' contextual vectors $\{x_t\}_{t=1}^T$ and the unknown regression model synthesized 
in the same way as in Section~\ref{sec:numerical}.
We also plot the consumer surplus for a hypothetical pricing algorithm that is \emph{not} subject to any privacy constraints
as dashed lines in Figure \ref{fig:surplus}.
Note that we did not incorporate consumers' surplus from the protection of their private data,
which is difficult to measure and compare against the surplus from their purchasing decisions.
{  As we can see from Figure \ref{fig:surplus}, as $\varepsilon$ increases from 0 to 1, the implied privacy protection becomes \emph{weaker}
as the adversary has a stronger ability to distinguish between neighboring databases. This means that as $\varepsilon$ increases, the seller has less ability to discriminate against customers based on their personal data and features, resembling a transition from first-degree to third-degree price discrimination.
As a result, the consumer surplus increases as $\varepsilon$ increases and the seller extracts less of the consumer surplus from his/her limited ability to carry out price discrimination.}


\section{Conclusions and future directions}
\label{sec:con}

In this paper, we investigate how to protect the privacy of a customer's personal information and purchasing decisions in personalized dynamic pricing with demand learning.  Under the generalized linear model of the demand function, we propose a privacy-preserving constrained MLE policy. We establish both the privacy guarantee under the notion of anticipating differential privacy (DP) and the regret bounds for oblivious adversarial and stochastic settings.

There are several future directions. First, we could extend the current privacy setting to the local DP \citep{Evfimievski2003limiting,kasiviswanathan2011can}, which is a stronger notion of DP. The local DP is suitable for distributed environments, as user terminals need to randomize data before sending it to the center. A very recent paper by \cite{ren2020mab} investigates the UCB algorithm under the local DP. It would be interesting to study the personalized dynamic pricing under this stronger notion of DP. More importantly, as privacy has become a significant concern for the public, especially in the e-commerce domain, we believe that systematic research on privacy-preserving revenue management will become increasingly important in both academia and industry. 
While there is relatively less research in this area, we hope our work inspires future studies on privacy-aware operations management (e.g., inventory control or assortment optimization) based on the DP framework.

\section*{Acknowledgment}
The authors thank the department editor, the associated editor, and the anonymous referees for many useful suggestions and feedback, which greatly improves the paper. Xi Chen is supported by National Science Foundation via the Grant IIS-1845444 and Facebook Faculty Research Award. David Simchi-Levi is supported by MIT-Accenture Alliance for Business Analytics and the MIT Data Science Lab.

\bibliography{refs}
\bibliographystyle{apa-good}

\end{document}